%% file: main.tex
\RequirePackage{etex}
\documentclass[11pt,letterpaper]{article} % following FOCS instruction, added 11pt and letterpaper

\input{package}
\input{macro}

\usepackage{libertine}
\usepackage{fullpage} % 1 inch all around by default
\usepackage{mydefs}

\newcommand{\DHSparsify}{{\sc DH-Sparsify}}
\newcommand{\DHOnestep}{{\sc DH-Onestep}}
\newcommand{\CoresetFinder}{{\sc CoresetFinder}}
\newcommand{\UHSparsify}{{\sc UH-Sparsify}}
\newcommand{\UHOnestep}{{\sc UH-Onestep}}

\newcommand{\Tend}{i_{\mathrm{end}}}
\newcommand{\sfn}{g} % submodular function. f is used as a hyperarc
\newcommand{\js}{s} % lower bound proof index
\newcommand{\GH}{G}

\title{Nearly Tight Spectral Sparsification of Directed Hypergraphs\\ by a Simple Iterative Sampling Algorithm}

\author{%
\makebox[1.0\linewidth]{Kazusato Oko}\\
The University of Tokyo \\ Center for Advanced Intelligence Project, RIKEN \\
\href{mailto:oko-kazusato@g.ecc.u-tokyo.ac.jp}{oko-kazusato@g.ecc.u-tokyo.ac.jp}
\and
Shinsaku Sakaue \\
The University of Tokyo \\
\href{mailto:sakaue@mist.i.u-tokyo.ac.jp}{sakaue@mist.i.u-tokyo.ac.jp}
\and
Shin-ichi Tanigawa \\
The University of Tokyo \\
\href{mailto:tanigawa@mist.i.u-tokyo.ac.jp}{tanigawa@mist.i.u-tokyo.ac.jp}
}

\date{}

\begin{document}

\maketitle

\begin{abstract}
  Spectral hypergraph sparsification, an attempt to extend well-known spectral graph sparsification to hypergraphs, has been extensively studied over the past few years. For undirected hypergraphs, Kapralov, Krauthgamer, Tardos, and Yoshida~(2022) have proved an $\varepsilon$-spectral sparsifier of the optimal $O^*(n)$ size, where $n$ is the number of vertices and $O^*$ suppresses the $\varepsilon^{-1}$ and $\log n$ factors. For directed hypergraphs, however, the optimal sparsifier size has not been known. Our main contribution is the first algorithm that constructs an $O^*(n^2)$-size $\varepsilon$-spectral sparsifier for a weighted directed hypergraph. Our result is optimal up to the $\varepsilon^{-1}$ and $\log n$ factors since there is a lower bound of $\Omega(n^2)$ even for directed graphs. We also show the first non-trivial lower bound of $\Omega(n^2/\varepsilon)$ for general directed hypergraphs. The basic idea of our algorithm is borrowed from the spanner-based sparsification for ordinary graphs by Koutis and Xu~(2016). Their iterative sampling approach is indeed useful for designing sparsification algorithms in various circumstances. To demonstrate this, we also present a similar iterative sampling algorithm for undirected hypergraphs that attains one of the best size bounds, enjoys parallel implementation, and can be transformed to be fault-tolerant.
\end{abstract}

\section{Introduction}\label{sec:introduction}

Graph sparsification is a fundamental idea for developing efficient algorithms and data structures.
One of the earliest developments in this context is a cut sparsifier due to Bencz{\'u}r~and~Karger~\citep{benczur1996approximating}, which approximately keeps the size of cuts (by adjusting edge weights).
Spielman~and~Teng~\citep{spielman2011spectral} introduced a generalized notion called a spectral sparsifier, which approximately preserves the spectrum of the Laplacian matrix of a given graph.
Since this seminal work, spectral sparsification of graphs has been extensively studied and used in many applications.
See, e.g., \citep{vishnoi2013Lxb,teng2016scalable,spielman2019spectral} for more details on spectral graph sparsification.

This paper studies spectral sparsification of undirected/directed hypergraphs.
A hypergraph is a standard tool for generalizing graph-theoretic arguments in a set-theoretic setting, and extending a theory for graphs to hypergraphs is a common theoretical interest.
Besides, many hypergraph-based methods \citep{hein2013total,yadati2019hypergcn,takai2020hypergraph,yadati2020nhp,zhang2020re} have recently been attracting much attention as extensions of graph-based methods, which also increases the demand for advancing the theory of spectral hypergraph sparsification.

An {\em undirected hypergraph} is defined by a tuple $H = (V, F, z)$, where $V$ is a finite vertex set, $F$ is a set of subsets of $V$, and $z\colon F\to\R_+$.
Each element in $F$ is called a hyperedge and $z_f:=z(f)$ is called the weight of $f\in F$ in $H$.
The Laplacian $L_H\colon \R^V\to\R^V$  of $H$  is defined as a nonlinear operator such that
\begin{align}
  x^\top L_H(x) = \sum_{f\in F} z_f \max_{u, v\in f}(x_u-x_v)^2
  \quad
  \text{for all $x\in \R^V$.}
\end{align}
If $x$ is restricted to $\{0,1\}^V$, $x^\top L_H(x)$ represents the cut function of $H$. In this sense, the above definition gives a proper extension of the ordinary graph Laplacian.
(Here, $x^\top L_H(x)$ is an abuse of notation since $L_H(x)$ is not defined uniquely; nevertheless, this notation is widely used in analogy to the case of ordinary graphs.)

A {\em directed hypergraph}  $H = (V, F, z)$ consists of a finite set $V$, a set $F$ of hyperarcs, and $z\colon F\ni f\mapsto z_f\in \R_+$,
where each {\em hyperarc} $f\in F$ is a pair $(t(f), h(f))$ of non-empty subsets of $V$, called the \textit{tail} and the \textit{head} (which may not be disjoint).
The Laplacian $L_H\colon \R^V\to\R^V$  of $H$ is defined as a nonlinear operator such that
\begin{align}
  x^\top L_H(x) = \sum_{f\in F} z_f \max_{u\in t(f), v\in h(f)}(x_u-x_v)^2_+
  \quad
  \text{for all $x\in \R^V$,}
\end{align}
where $(\cdot)_+ = \max\set*{\cdot, 0}$ (and $(\cdot)_+^2 = \prn*{\max\set*{\cdot, 0}}^2$).
If $x \in \set*{0,1}^V$, the definition of directed hypergraph Laplacian $L_H$ also captures the cut function of $H$.
Importantly, cut functions of directed hypergraphs can represent a large class of submodular functions~\citep{fujishige2001realization}.\footnote{In fact, any set function can be represented as a cut function of some directed hypergraph if negative weights are allowed~\citep{fujishige2001realization}.}
Directed hypergraphs are also useful for modeling higher-order directional relations that appear in, e.g., propositional logic \citep{gallo1993directed} and causal inference~\citep{javidian2020on}, which have constituted a motivation for studying spectral properties of directed hypergraphs~\citep{chan2019diffusion}.

Given an undirected/directed hypergraph $H = (V, F, z)$ and $\eps \in (0, 1)$, a hypergraph $\Htl = (V, \Ftl, \ztl)$ is called an \textit{$\eps$-spectral sparsifier} of $H$ if it satisfies $\Ftl \subseteq F$ and
\begin{align}
  \label{eq:Intro-Def-of-Sparsifier-hyper}
  (1-\eps)x^\top L_H (x) \leq x^\top L_{\Htl} (x) \leq (1+\eps)x^\top L_H (x)
  \quad
  \text{for all $x \in \R^V$.}
\end{align}
One of the big motivations for studying spectral sparsification of directed hypergraphs comes from the connection to the representation of submodular functions.
Since such a cut-function representation uses $\Omega(2^{|V|})$ hyperarcs in general, a spectral sparsifier of a directed hypergraph can serve as a compact approximate representation (see \cref{sec:submodular} for more details).

Soma~and~Yoshida~\citep{soma2019spectral} initiated the study of spectral hypergraph sparsification and gave an algorithm for constructing an $\eps$-spectral sparsifier with $O(n^3\log n/\eps^2)$ hyperedges, where $n$ is the number of vertices.
Unlike ordinary graphs, the hypergraph size can be as large as $2^n$ (and $4^{n}$ if directed).
Thus, obtaining a polynomial bound is already nontrivial.
For undirected hypergraphs, the result by Soma~and~Yoshida~\citep{soma2019spectral} has been improved to $\Otl(nr^3/\eps^2)$~\citep{bansal2019new} and to $\Otl(nr/\eps^{O(1)})$~\citep{kapralov2021towards},\footnote{We use $\Otl$ to hide $\poly(\log (n/\eps))$ factors.} where $r$ denotes the maximum size of a hyperedge in the input hypergraph $H$ and is called the {\em rank} of $H$.
Kapralov~et~al.~\citep{kapralov2021spectral} has removed the dependence on $r$ and obtained a nearly linear bound of $\Otl(n/\eps^4)$.
Very recently, an improved bound of $\Otl(n/\eps^2)$ has been shown in~\cite{Jambulapati2022-qh,Lee2022-rp} (concurrently to our work).
This upper bound is nearly tight since the $\Omega(n/\eps^2)$ lower bound applies even to ordinary graphs \citep{andoni2016sketching,carlson2019optimal}.

As for spectral sparsification of directed hypergraphs, Soma~and~Yoshida~\citep{soma2019spectral} showed that their algorithm is also applicable, and hence the $O(n^3\log n/\eps^2)$ bound also holds for directed hypergraphs.
Later, Kapralov~et~al.~\citep{kapralov2021towards} gave an $O(n^2r^3\log^2 n/\eps^2)$ bound for unweighted directed hypergraphs, where the rank $r$ is defined by $r = \max_{f \in F} \set*{|h(f)| + |t(f)|}$ in the directed case.
Recently, for the case of cut sparsification,
Rafiey~and~Yoshida~\citep{rafiey2022sparsification} obtained sparsifiers with $O(n^2r^2/\eps^2)$ hyperarcs.\footnote{This bound follows from their general result on sparsification of submodular functions.}
%\footnote{This bound is obtained by applying their general bound for sparsification of submodular functions to cut functions.}
See \cref{table:Intro-previous-studies-directed}.
On the other hand, a well-known $\Omega(n^2)$ lower bound for directed graphs \citep{cohen2017almost} is valid for directed hypergraphs.
Therefore, a central open question in this context is: \textit{can we obtain an upper bound of $\tilde{O}(n^2/\eps^{O(1)})$ that has no dependence on the rank $r$?}

\subsection{Main Results and Idea}\label{subsec:our-results}

Our main contribution is the first algorithm that constructs an $\eps$-spectral sparsifier with $\Otl(n^2/\eps^2)$ hyperarcs for a directed hypergraph, thus settling the aforementioned question.

\begin{table}[tb]
    \centering
    \caption{Bounds on sparsification of directed hypergraphs. In the time complexity, additive $\poly(n, 1/\eps)$ terms are omitted. Note that Kapralov~et~al.~\citep{kapralov2021towards} assume the unweighted case.}
    \label{table:Intro-previous-studies-directed}
    \begin{tabular}{cccc} \toprule
    Method & Cut/Spectral & Bound & Time complexity\\ \midrule
    Soma~and~Yoshida~\citep{soma2019spectral} & Spectral & $O(n^3\log n/\eps^2)$ & $O(mr^2)$  \\
    Kapralov~et~al.~\citep{kapralov2021towards} & Spectral & $O(n^2r^3\log^2 n/\eps^2)$ & $O(mr^2)$  \\
    Rafiey~and~Yoshida~\citep{rafiey2022sparsification} & Cut & $O(n^2r^2/\eps^2)$ & $O(m2^r)$  \\
    This paper & Spectral & $O(n^2 \log^3 (n/\eps)/\eps^2)$ & $O(mr^2)$  \\
    \bottomrule
    \end{tabular}
\end{table}

%\begin{theorem}\label[theorem]{theorem:Intro-DH-main}
%  Given a directed hypergraph with $n$ vertices and $\eps \in (0, 1)$, there is a polynomial-time algorithm that returns an $\eps$-spectral sparsifier with $O\prn*{\frac{n^2}{\eps^2}\log^3 \frac{n}{\eps}}$ hyperarcs with probability at least $1 - O\prn*{\frac{1}{n}}$.
%\end{theorem}

\begin{theorem}\label[theorem]{theorem:DHS-main}
  Let $H=(V,F,z)$ be a directed hypergraph with $n$ vertices.
  For any $\eps \in (0, 1)$, our algorithm (shown in \cref{alg:DHS-iterative}) returns an $\epsilon$-spectral sparsifier $\Htl=(V,\tilde{F},\tilde{z})$ of $H$ such that $|\tilde{F}| = O\left(\frac{n^2}{\eps^2}\log^3\frac{n}{\eps}\right)$ with probability at least $1-O\left(\frac{1}{n}\right)$.
  Its time complexity is $O(mr^2)$ with probability at least $1-O\left(\frac{1}{n}\right)$, where $m = |F|$ and $r$ is the rank of $H$.
\end{theorem}

This bound improves the previous results and is optimal up to the $\eps^{-1}$ and logarithmic factors due to the presence of the $\Omega(n^2)$ lower bound for directed graphs.
We prove \cref{theorem:DHS-main} in \cref{chapter:DHS} by providing a concrete algorithm and its analysis.

A natural next question would be whether the $\eps^{-1}$ term can be deleted.
Our new lower bound shows that the $\eps^{-1}$ term is indeed necessary, and an $\eps$-spectral sparsifier of size $O(n^2)$ may not exist in general, thus complementing our upper bound.

\begin{theorem}\label[theorem]{theorem:Intro-DH-lower}
Let $n \in \Z_{>0}$.
For any $\eps \in \prn*{\frac{1}{4n}, 1}$, there is a directed hypergraph $H=(V,F,z)$ with $2n$ vertices, $\Omega \prn*{\frac{n^2}{\eps}}$ hyperarcs, and the rank three that has no sub-hypergraph $\Htl=(V, \Ftl, \ztl)$ such that $\Ftl\subsetneq F$ and $(1-\eps )x^\top L_{H}(x) \leq x^\top L_{\Htl}(x)\leq(1+\eps )x^\top L_{H}(x)$ for all $x\in \{0,1\}^V$.
\end{theorem}
This gives a lower bound even for the case of cut sparsification and is the first nontrivial lower bound for sparsification of directed hypergraphs.
We present the proof of \cref{theorem:Intro-DH-lower} in \cref{section:LowerBound}.

The basic idea of our algorithm for \cref{theorem:DHS-main} comes from a spanner-based sparsification method for undirected graphs by Koutis~and~Xu~\citep{koutis2016simple}, in contrast to the method of \citep{kapralov2021spectral} for nearly tight sparsification of undirected hypergraphs.
The analysis of \citep{kapralov2021spectral} uses a technique called \emph{weight assignment}~\cite{chen2020near}, which crucially depends on linear algebraic arguments on the linear Laplacian of some underlying undirected graph.
\emph{Directed} hypergraphs, however, do not have such convenient underlying \emph{undirected} graphs, and hence their idea cannot be utilized.
We thus take an alternative route and use the algorithmic framework of Koutis~and~Xu~\citep{koutis2016simple}---iteratively select important edges and sample the remaining edges.
Due to its combinatorial nature, we can analyze errors via combinatorial arguments instead of linear algebraic tools.
Although our algorithm is as simple as theirs, our analysis for proving \cref{theorem:DHS-main} involves novel techniques.
Specifically, while building on a recent chaining-based analysis~\citep{kapralov2021towards,kapralov2021spectral}, we develop a completely new \emph{discretization scheme} based on a non-trivial combinatorial observation to obtain the optimal upper bound.
See \cref{section:DHS-TechnicalOverview} for an overview of our analysis.

\subsection{Additional Results}

\paragraph{Undirected hypergraph sparsification.}
The iterative sampling approach mentioned above indeed has much potential in hypergraph sparsification.
In \cref{section:UHS}, we exhibit its power by presenting a natural extension of the spanner-based algorithm by Koutis~and~Xu~\citep{koutis2016simple} to undirected hypergraphs.
The concept of spanners in graphs can be naturally extended to undirected hypergraphs, and accordingly, Koutis and Xu's algorithm can also be extended to undirected hypergraphs.
Based on a result by Bansal~et~al.~\citep{bansal2019new}, we show that the resulting algorithm constructs an $\eps$-spectral sparsifier with $O\prn*{\frac{nr^3}{\eps^2}\log^2 n}$ hyperedges, which is nearly optimal if $r$ is constant and matches the bound of \citep{bansal2019new} (up to a $\log n$ factor).
Moreover, our algorithm inherits advantages of the spanner-based approach in that it can be implemented in parallel~\citep{koutis2016simple} and can be converted to be fault-tolerant~\citep{zhu2019improved}, demonstrating that the iterative sampling approach can enjoy various useful extensions.

\paragraph{Application to learning of submodular functions.}
A notable application of directed hypergraph sparsification due to \citep{soma2019spectral} is agnostic learning of submodular functions.
In \cref{sec:submodular}, we apply our method to this setting and obtain an $\tilde{O}\prn*{ \frac{n^3}{\eps^4} + \frac{1}{\eps^2}\log\frac{1}{\delta} }$ sample complexity bound for agnostic learning of nonnegative hypernetwork-type submodular functions on a ground set of size $n$, improving the previous $\tilde{O}\prn*{ \frac{n^4}{\eps^4} + \frac{1}{\eps^2}\log\frac{1}{\delta} }$ bound in~\citep{soma2019spectral}.
Note that since the rank $r$ of a hypergraph representing a submodular function can be $O(n)$, eliminating the dependence on $r$ in the sparsifier size (i.e., our improvement from~\citep{kapralov2021towards}) is crucial in this application.
It should be mentioned that this application only requires cut sparsifiers.
Nevertheless, since our result gives the first near-optimal bound even on the size of cut sparsifiers of directed hypergraphs, this application serves as a good motivation for our result.

\subsection{Related Work}
Besides the aforementioned application to agnostic learning of submodular functions, there are many other potential applications that involve the quadratic form $x^\top L_H(x)$ (which is sometimes called the {\em energy} of hypergraphs), e.g., clustering~\citep{takai2020hypergraph}, semi-supervised learning~\citep{hein2013total,yadati2019hypergcn,zhang2020re,Li2020-db}, and link prediction~\citep{yadati2020nhp}.
For example, Li et al.~\citep{Li2020-db} use the quadratic form as a smoothness regularizer.
Our result on spectral sparsification can be useful when dealing with such regularizers on dense directed hypergraphs.

Cohen et al.~\cite{cohen2017almost} studied directed graph sparsification under a different definition of approximation based on \textit{Eulerian scaling}.
While their definition is compatible with fast Laplacian solvers, how to extend it to directed hypergraphs seems non-trivial.
Our definition is based on a general notion called \textit{submodular transformations}~\citep{Yoshida2019-ez} and admits a natural interpretation as a generalization of cut sparsification of directed hypergraphs.

\section{Preliminaries}\label{section:preliminaries}

We usually denote a directed hypergraph by $H=(V, F, z)$, the numbers of vertices by $n$, and the numbers of hyperarcs by $m$.
The Laplacian $L_H\colon \R^V \to \R^V$ is defined as a nonlinear operator that satisfies $x^\top L_H(x) = \sum_{f \in F}z_f\max_{u \in t(f), v \in h(f)} (x_u - x_v)_+^2$ for all $x \in \R^V$, where $h(f), t(f) \subseteq V$ are the head and the tail of $f$, respectively.
For each $f \in F$, we denote the contribution of $f$ to $x^\top L_H(x)$ by $Q_H^x(f) = z_f\max_{u\in t(f),v\in h(f)}(x_u-x_v)_+^2$, which we call the \textit{energy} of $f$.
Note that $x^\top L_H(x) = \sum_{f \in F} Q_H^x(f)$ holds.
For any subset $F'$ of $F$, we let $Q_H^x(F') = \sum_{f \in F'} Q_H^x(f)$, i.e., the sum of energies over $F'$.
For a hyperarc $f \in F$, we define its \textit{biclique} as an arc set $C(f) = \Set*{(u, v)}{u\in t(f), v\in h(f)}$.
For a subset $F'\subseteq F$, we let $C(F') = \bigcup_{f \in F'}C(f)$.
Below, we often take $\argmax_{f \in F'} \zeta(f)$ for a function $\zeta\colon F\to \R$ and a hyperarc subset $F'\subseteq F$.
For convenience, we let such argmax (or argmin) operations always return a singleton by using some tie-breaking rule with a pre-defined total order on $F$.
For example, if vertices are labeled by $1,\dots,n$ and each $f \in F$ is labeled by vertices in $f$, we may use the lexicographical order on $F$ with respect to the labels.
Similarly, we break ties when taking argmax/argmin on any $E' \subseteq V\times V$.
We will often use the following Chernoff bound.

\begin{proposition}[\citep{alon2016probabilistic}]\label[proposition]{theorem:chernoff}
Let $X_1,X_2,\cdots,X_m$ be independent random variables in the range of $[0,a]$.
For any $\delta \in [0,1]$ and $\mu \geq \mathbb{E}\left[\sum_{i=1}^{m}X_i\right]$, we have
    \begin{align}
        \nonumber
        \mathbb{P}\left[
            \left|
                \sum_{i=1}^{m}X_i - \mathbb{E}\left[\sum_{i=1}^{m}X_i\right]
            \right|
            > \delta \mu
        \right]
        \leq
        2\exp
        \left(
            -\frac{\delta^2\mu}{3a}
        \right)
        .
    \end{align}
\end{proposition}

\section{Technical Overview}\label{section:DHS-TechnicalOverview}
Our algorithm is an iterative algorithm whose each step goes as follows:
given a hypergraph $H=(V,F,z)$ from a previous iteration, it constructs a set $S$ of heavy hyperarcs, called a \textit{coreset}, which is kept deterministically in this step, and samples the remaining hyperarcs with probability $1/2$, where weights of sampled ones are doubled.
This single step yields a hypergraph with fewer hyperarcs, which is taken as input in the next step.
We iterate this until a sub-hypergraph of the desired size is obtained.
Roughly speaking, the size of the coreset is about $\Otl(n^2/\eps^2)$, and after about $O(\log\prn*{{m\eps^2}/{n^2}})$ iterations, we obtain a sub-hypergraph of size $\Otl(n^2/\eps^2)$.
This algorithmic framework is identical to that of Koutis~and~Xu~\citep{koutis2016simple} for ordinary undirected graph sparsification, which iteratively constructs a bundle of spanners (instead of a coreset) and sample the remaining edges with probability $1/4$.
% In each iteration, we keep sampled hyperarcs and a coreset of size about $\Otl(n^2/\eps^2)$, which plays a similar role to a bundle of spanners in Koutis~and~Xu~\citep{koutis2016simple}.
% Thus, if a hypergraph with $m$ hyperarcs is given, the size of the resulting hypergraph is concentrated around $\Otl(n^2/\eps^2)+m/2$, which is at most $\frac{3}{4}m$ if $\Otl(n^2/\eps^2) \le \frac{m}{4}$ (otherwise the initial $H$ is already sparsified).
% Hence, after about $\log_{4/3}\prn*{{m\eps^2}/{n^2}}$ iterations, we obtain a sub-hypergraph of size $\Otl(n^2/\eps^2)$.

We then describe how to analyze the sparsification error.
Note that if a sub-hypergraph produced in each step is a sparsifier of a hypergraph $H=(V,F,z)$ given from the previous step with a sufficiently large probability, then we can bound the error accumulated over the iterations.
Thus, we focus on the analysis of a single step (which is presented in \cref{lemma:DHS-onestep}).
To bound the sparsification error in $Q_H^x(F) = x^\top L_H(x)$ for all $x \in \R^V$ in each step, we adopted a chaining-type argument~\citep{kapralov2021towards,kapralov2021spectral};
this enables us to derive a desired uniform bound on a continuous domain from a pointwise bound via adaptive scaling of the domain discretization.
Here, how to design a discretization scheme crucially affects how sharp the resulting uniform bound is.
Therefore, we need to design an appropriate discretization scheme by carefully looking at the structure of directed hypergraphs.

% We %then
% explain how to upper bound the sparsification error in the energy $Q_H^x(F) = x^\top L_H(x)$, which is nonlinear in $x$, for all $x \in \R^V$.
% By applying the Chernoff bound to the independent sampling of hyperarcs, it is easy to bound the sparsification error {\em for each} $x \in \R^V$.
% On the other hand, we want to bound the error in the total energy $Q_H^x(F) = \sum_{f \in F} Q_H^x(f)$ {\em uniformly for all} $x \in \R^V$.
% To bridge this gap, we adopted a chaining-type argument~\citep{kapralov2021towards,kapralov2021spectral}, which enables us to derive the desired uniform bound on a continuous domain from the union bound via appropriate discretization of the domain.
% The discretization scale is defined to be progressively smaller to control the trade-off between the discretization error and the number of discretized values.
% Therefore, how to design an appropriate discretization scheme is crucial for obtaining a sharp uniform bound.

We below sketch our discretization scheme.
Inspired by the previous studies \citep{kapralov2021towards,kapralov2021spectral}, we classify hyperarcs $f \in F\setminus S$ based on their energies $Q_H^x(f)$.
Here, since the coreset $S$ is always selected, we can exclude it when discussing the following probabilistic arguments.
For each $x\in \R^V$, we consider a partition of $F\setminus S$ into $F_i^x$ ($i\in \Z$) such that each $F_i^x$ consists of hyperarcs $f$ with energies $Q_H^x(f) \approx 2^{-i} Q_{H}^x(F)$.
Then, the Chernoff bound offers the following pointwise guarantee for each $x \in \R^V$:
\begin{align}
    \nonumber
    \mathbb{P}\left[
    |Q_{\tilde{H}}^x(\tilde{F}_i^x) - Q_{H}^x(F_i^x)|
    \geq
    \eps Q_{H}^x(F)
    \right]
    \lesssim
    \exp
    \left(
        - \frac{\varepsilon^2 Q_{H}^x(F)}{2^{-i}Q_{H}^x(F)}
    \right)
    =
    \exp
    \left(
        - \varepsilon^{2}2^i
    \right),
\end{align}
where $\Htl$ is a sparsifier obtained from $H$ and $Q_{\tilde{H}}^x(\tilde{F}_i^x)$ denotes the energy of $\Htl$ with hyperarcs restricted to $F_i^x$.
To obtain a desired uniform bound using this inequality, we need to design a discretization scheme that satisfies the following two requirements:
\begin{description}
    \item[(R1)] the discretization error is $O(\eps)$, and
    \item[(R2)] the number of possible discretized energies is bounded by about $\exp(\eps^2 2^i)$.
\end{description}
Kapralov~et~al.~\citep{kapralov2021towards} obtained such a scheme by looking at underlying clique digraphs.
By contrast, we obtain a discretization scheme by directly looking at hypergraphs.
This strategy enables us to eliminate the extra $r^3$ factor in their bound, but it also poses a new challenge.

We explain the challenge when designing such a discretization scheme by directly looking at hypergraphs.
Once $x \in \R^V$ is fixed, the number of hyperarcs $f$ with $Q_H^x(f) \approx 2^{-i} Q_{H}^x(F)$ is bounded by about $2^i$; on the other hand, we need to prepare at least $\poly(n, 1/\eps)$ possible discretized energies for each $f$ to satisfy requirement~(R1).
Thus, naive counting implies that the number of total discretized energies for all $f \in F_i^x$ is $(\poly(n, 1/\eps))^{2^i} \approx \exp(\Otl(2^i))$, which is too large to satisfy requirement~(R2).
To overcome this problem, we need an additional combinatorial idea:
we count the number of discretized energies by focusing on the number of possible critical pairs.
We say that $(u,v)\in C(F\setminus S)$ is a {\em critical pair} of $f$ if $(u, v) = \argmax_{u' \in t(f), v' \in h(f)} (x_{u'} - x_{v'})_+^2$ (see also \Cref{subfig:critical}).
Suppose that a lot of hyperarcs in $F_i^x$ share a common critical pair for a given $x\in \R^V$, particularly when $F_i^x$ contains as many as $2^i$ hyperarcs.
Then, since the energy of $f$ is determined by the $(x_{u} - x_{v})_+^2$ value of the critical pair $(u,v)$ of $f$, we may get a sharper bound on the number of discretized energies by defining a discretization scheme based on $(x_{u} - x_{v})_+^2$ values so that hyperarcs with the same critical pairs share the same discretized energies (up to scaling of weights).

To accomplish the counting based on this idea, we use the existence of a coreset kept in each iteration.
As we will see shortly from the definition, a $\lambda$-coreset $S \subseteq F$ contains $\lambda$ heaviest hyperarcs for each $(u, v) \in C(F)$ (see also \Cref{subfig:coreset}).
Roughly speaking, important properties of $\lambda$-coresets are as follows:
\begin{description}
    \item[(P1)] $|S| \le \lambda n^2$,
    \item[(P2)] for any fixed $x \in \R^V$, many hyperarcs with large energies are included in $S$, and
    \item[(P3)] for any fixed $x\in \R^V$, the number of critical pairs of hyperarcs in $F^x_i$ is at most $2^i/\lambda$.\footnote{For ease of exposition, $\lambda$ is used differently from \cref{chapter:DHS}. In \cref{section:DHS-JustificationForOneIteration}, we will instead define $F_i^x$ based on $2^{-i} Q_{H}^x(F) / \lambda$ values and, accordingly, bound the number of critical pairs by $2^i$ (\cref{lemma:DHS-onestep-MoreEnergy}).}
\end{description}
If we set $\lambda = \Otl(\eps^{-2})$, the size of $\lambda$-coreset $S$ is $\Otl(n^2/\eps^2)$ by property~(P1), which is small enough that the output size decreases geometrically in each iteration until we obtain an $\Otl(n^2/\eps^2)$ size sparsifier.
Property~(P2) bounds the range of $i$ such that $F_i^x$ is non-empty.
Most importantly, property~(P3) implies that if we count possible discretized energies over $F_i^x$, the total number is at most $\prn*{\poly(n, 1/\eps)}^{2^i/\lambda} \approx \exp\prn*{ \Otl(\eps^2 2^i) }$, satisfying requirement~(R2).

In summary, once the coreset is selected, we can categorize the remaining hyperarcs in each $F_i^x$ based on a moderate number of critical pairs, which yields a sharp bound on the number of possible discretized energies of the remaining hyperarcs.
This is the key idea of our discretization scheme, which, together with the chaining-type argument, provides the desired uniform bound on the sparsification error.

\section{Spectral Sparsification of Directed Hypergraphs}\label{chapter:DHS}

%The main result of this section is the following theorem.
%\begin{theorem}\label[theorem]{theorem:DHS-main}
%    Let $H=(V,F,z)$ be a directed hypergraph with $n$ vertices.
%    For any $\eps \in (0, 1)$, {\rm{\DHSparsify$(H,\eps)$}} given in \cref{alg:DHS-iterative} returns an $\epsilon$-spectral sparsifier $\Htl=(V,\tilde{F},\tilde{z})$ of $H$ such that $|\tilde{F}| = O\left(\frac{n^2}{\eps^2}\log^3\frac{n}{\eps}\right)$ with probability at least $1-O\left(\frac{1}{n}\right)$.
%    %\vee \frac{\log (n/\eps)}{n^2}\right)$.
%\end{theorem}

We prove \cref{theorem:DHS-main} by presenting a concrete algorithm.
%This section is organized as follows.
%\Cref{section:DHS-TechnicalOverview} overviews our technical ideas.
\Cref{section:DHS-AlgorithmDescription} presents our algorithm and key lemmas.
\Cref{section:DHS-JustificationForOneIteration} focuses on the analysis of a single iteration, and \Cref{section:DHS-NearlyTightSparsification} bounds the overall sparsification error and the resulting sparsifier size, thus proving \cref{theorem:DHS-main}.
\Cref{section:DHS-ComputationalComplexity} shows the $O(mr^2)$ time complexity bound of our algorithm.

\subsection{Algorithm Description}\label{section:DHS-AlgorithmDescription}
Our algorithm consists of {\CoresetFinder} (\cref{alg:DHS-Spanner}),  {\DHOnestep} (\cref{alg:DHS-onestep}), and {\DHSparsify} (\cref{alg:DHS-iterative}).
{\DHSparsify} iteratively calls {\DHOnestep}, which uses {\CoresetFinder} as a subroutine.
We below explain them one by one.

\begin{algorithm}[htb]
    \caption{\CoresetFinder$(H,\lambda)$: greedy algorithm for coreset construction.}\label{alg:DHS-Spanner}
    \begin{algorithmic}[1]
        \Require $H = (V, F, z)$ and $\lambda > 0$
        \Ensure $S \subseteq F$
        \State $S \gets \emptyset$ and $S^{uv}\gets \emptyset$ for each $(u,v)\in C(F)$
        \State $A^{uv}\gets \Set*{f \in F}{ (u, v)\in C(f)}$ for each $(u,v)\in C(F)$
        \For{each $(u,v) \in C(F)$}
        \If{$|A^{uv}\setminus S| \ge \lambda$}
        \State Find the first $\lambda$ heaviest hyperarcs $f_1^{uv},f_2^{uv},\cdots, f_{\lambda}^{uv} \in A^{uv}\setminus S$
        \State{Add $f_1^{uv},f_2^{uv},\cdots, f_{\lambda}^{uv}$ to $S^{uv}$}
        \Else
        \State $S^{uv} \gets A^{uv} \setminus S$
        \EndIf
        \State $S \gets S \cup S^{uv}$
        \EndFor
        \State \Return $S$
    \end{algorithmic}
\end{algorithm}
\begin{algorithm}[h]
  \caption{\DHOnestep$(H,\lambda)$: sampling algorithm called in each iteration in \cref{alg:DHS-iterative}.}\label{alg:DHS-onestep}
  \begin{algorithmic}[1]
      \Require $H = (V, F ,z)$ and $\lambda > 0$
      \Ensure $\Htl = (V, \Ftl, \ztl)$
          \State{$S \gets \text{\CoresetFinder}(H,\lambda)$} {\label[line]{line:dhone-findcoreset}}
          \State $\Ftl \gets S$ and $\ztl_f \gets z_f$ for $f \in S$
          \For{each $f \in F\setminus S$}
          \State{With probability $\frac12$, add $f$ to $\tilde{F}$ and set $\tilde{z}_f \gets 2z_f$}
          \EndFor
          \State \Return $\Htl = (V,\tilde{F}, \tilde{z})$
  \end{algorithmic}
\end{algorithm}
\begin{algorithm}[h]
  \caption{\DHSparsify$(H, \eps)$: iterative algorithm that computes an $\eps$-spectral sparsifier.}\label{alg:DHS-iterative}
  \begin{algorithmic}[1]
      \Require $H = (V,F,z)$ with $|V|=n$ and $|F|=m$, and $\eps > 0$
      \Ensure $\Htl = (V,\Ftl, \ztl)$
      \State $m^* \gets \frac{n^2}{\eps^2}\log^3\frac{n}{\eps}$ \Comment{This is the (asymptotic) target size of the resulting sparsifier.}
      \State $T \gets \ceil*{\log_{4/3} \left(\frac{m}{m^*}\right)}$
      \State $i \gets 0$, $\Htl_0 = (V, \Ftl_0, \ztl_0) \gets H$, and $m_0 \gets |\Ftl_0|$
      \While{$i < T$ and $m_i \geq \Cc m^*$} \Comment{$\Cc$ is a constant that is explained in \cref{section:DHS-NearlyTightSparsification}.}
      \State $\eps_i \gets \frac{\eps}{4 \log_{4/3}^2\left(\frac{m_i}{m^*}\right)}$ and $\lambda_i \gets \ceil*{ \frac{\Ca \log^3 m_i}{\eps_i^2} }$ {\label[line]{line:DHSparsify-eps-lambda}}
      \Comment{$\eps_i$ is used in the analysis.}
      \State $\Htl_{i+1} = (V,\tilde{F}_{i+1},\tilde{z}_{i+1}) \gets \text{\DHOnestep}(\Htl_{i}, \lambda_i)$
      \State $m_{i+1} \gets |\tilde{F}_{i+1}|$
      \State $i\gets i+1$
      \EndWhile
      \State $\Tend\gets i$ and $\Htl\gets \Htl_{\Tend}$
      \State \Return $\Htl = (V,\tilde{F}, \tilde{z})$
  \end{algorithmic}
\end{algorithm}

The first building block of our algorithm is {\CoresetFinder}$(H,\lambda)$ given in \cref{alg:DHS-Spanner}.
It takes a hypergraph $H$ and a parameter $\lambda$ as input, constructs a set, $S^{uv}$, of up to $\lambda$ hyperarcs for each $(u, v) \in C(F)$, and outputs $S = \bigcup_{(u, v) \in C(F)} S^{uv}$.
For each pair $(u,v)$ (in arbitrary order), $S^{uv}$ is obtained by selecting up to the $\lambda$ heaviest hyperarcs $f$ with $(u,v)\in C(f)$ among those not selected yet.
The parameter $\lambda$ controls the size of output $S$.

\begin{lemma}\label[lemma]{lemma:coreset}
  Let $H$ be a directed hypergraph and $\lambda$ be a positive integer.
  {\CoresetFinder}$(H, \lambda)$ returns a set $S$ of at most $\lambda n^2$ hyperarcs that can be partitioned into disjoints subsets $\Set*{S^{uv}}{(u,v)\in C(F)}$ satisfying the following conditions:
  \begin{enumerate}
      \item for any $(u,v)\in C(F)$, every $f\in S^{uv}$ satisfies $(u,v)\in C(f)$,
      \item if $(u,v)\in C(F\setminus S)$, $|S^{uv}|= \lambda$ holds, and
      \item for any $(u,v)\in C(F)$, $f\in S^{uv}$, and $f'\in F\setminus S$ such that $(u,v)\in C(f')$, $z_f\geq z_{f'}$ holds.
  \end{enumerate}
\end{lemma}
\begin{proof}
    Since {\CoresetFinder}$(H, \lambda)$ constructs $S^{uv}$ for each $(u,v)\in C(F)$ by selecting up to the $\lambda$ heaviest hyperarcs $f$ with $(u,v)\in C(f)$ among those that have not been selected yet, $S^{uv}$ for $(u, v) \in C(F)$ are mutually disjoint.
    This also implies $|S| = \sum_{(u, v) \in C(F)}|S^{uv}| \le \lambda n^2$ and the first and third conditions.
    After $S$ is constructed, if there is a hyperarc $f' \in F \setminus S$ such that $(u, v) \in C(f')$, then $\lambda$ hyperarcs must have been added to $S^{uv}$.
    Hence $|S^{uv}| = \lambda$ if $(u,v) \in C(F \setminus S)$, implying the second condition.
\end{proof}
We call the set $S$ shown in \cref{lemma:coreset} a \emph{coreset}, which plays a key role in the analysis.
\begin{definition}
Given a directed hypergraph $H=(V,F,z)$, a subset $S\subseteq F$, and a positive integer $\lambda$, we say $S$ is a {\em $\lambda$-coreset} of $H$ if $S$ can be partitioned into disjoints subsets $\Set*{S^{uv}}{(u,v)\in C(F)}$ satisfying the three conditions in the statement of \cref{lemma:coreset}.
\end{definition}
In short, if there is a hyperarc $f' \notin S$ with $(u, v) \in C(f')$, $S^{uv}$ contains (at least) $\lambda$ hyperarcs that are at least as heavy as $z_{f'}$.
\Cref{subfig:coreset} illustrates an example of a coreset.
We use this coreset as a counterpart of a bundle of spanners in the spanner-based sparsification.

\begin{figure}[tb]
  \label{fig:DHS-coreset}
  \centering
  \begin{minipage}[b]{.3302\linewidth}
      \centering
      \includegraphics[width=\linewidth]{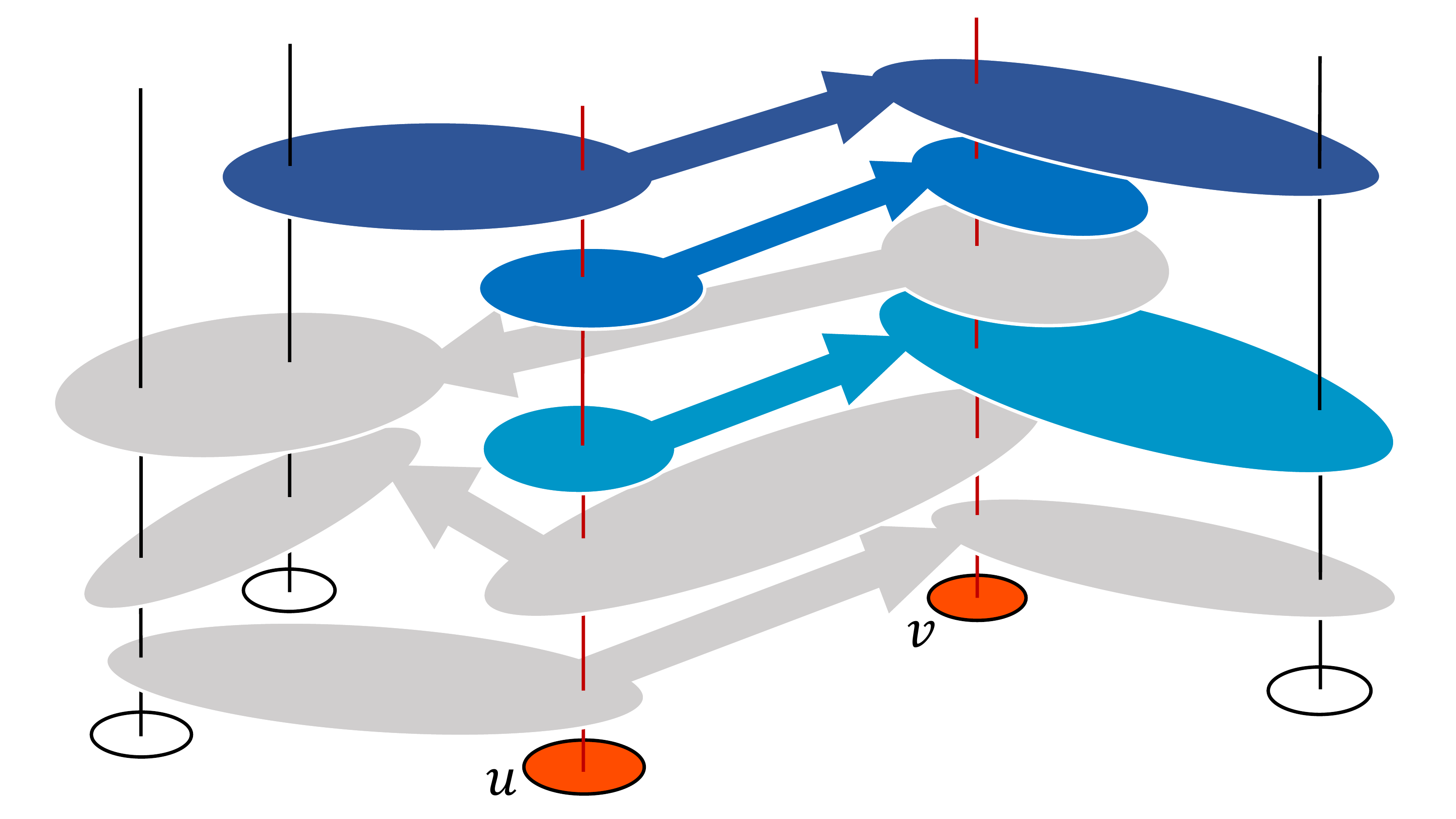}
      \subcaption{Coreset}
      \label{subfig:coreset}
  \end{minipage}
  \hspace{30pt}
  \begin{minipage}[b]{.3302\linewidth}
      \centering
      \includegraphics[width=\linewidth]{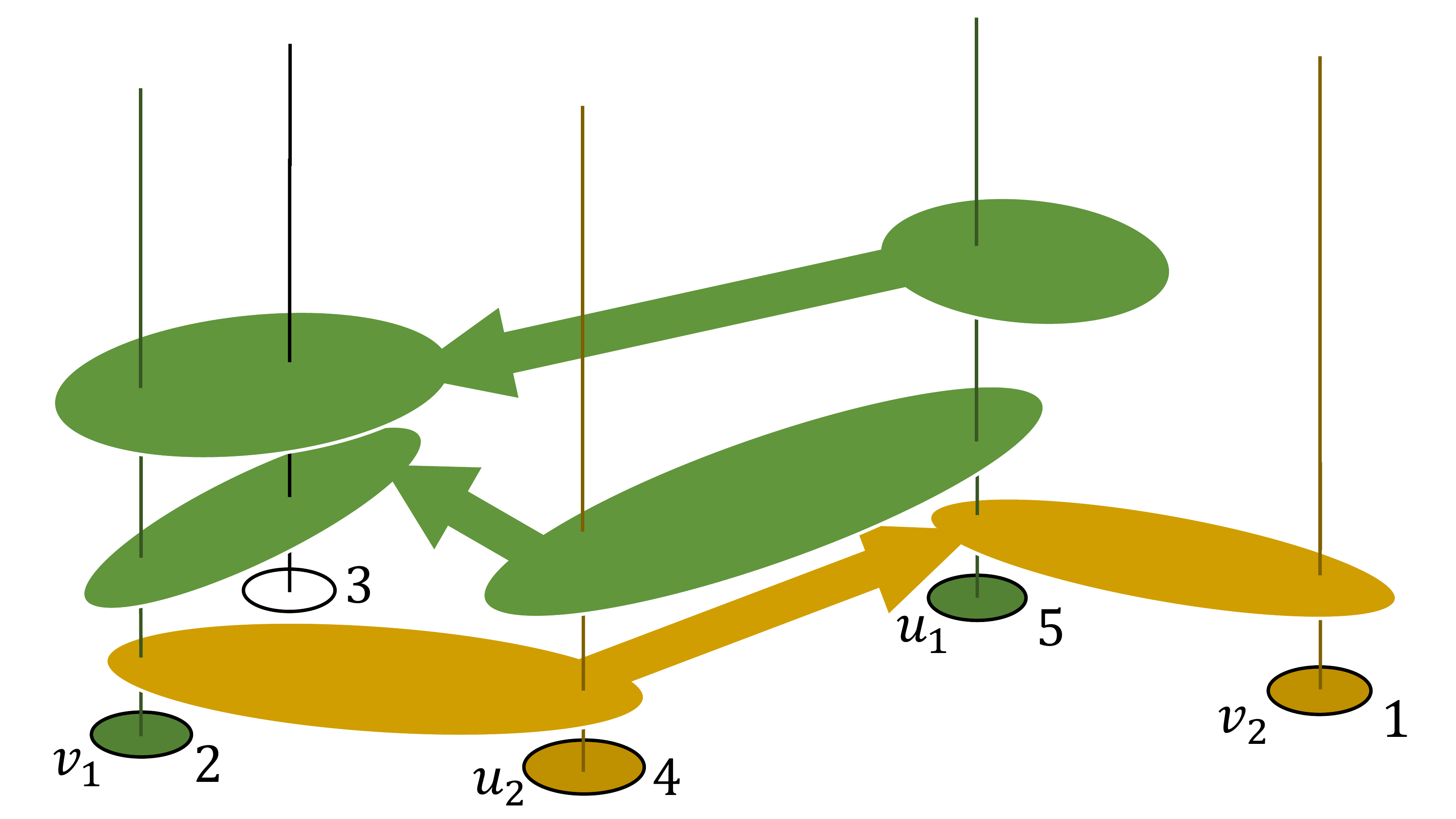}
      \subcaption{Critical pairs}
      \label{subfig:critical}
  \end{minipage}
  \caption{
      Illustration of a coreset and critical pairs on (a part of) a given hypergraph.
      A circle is a vertex, and a hyperarc is indicated by an arrow and two ellipses representing a head and a tail.
      A hyperarc contains a vertex if the line originating from the vertex pierces its head or tail.
      \Cref{subfig:coreset} presents an image of a coreset, focusing on a vertex pair $(u, v)$.
      Suppose that the hyperarcs are aligned in decreasing order of their weights from top to bottom.
      The blue hyperarcs are the three heaviest ones having $u$ and $v$ as elements of their tails and heads, respectively, and they are included in a subset $S^{uv}$ of a coreset $S$.
      We suppose that gray hyperarcs are not in $S$.
      While the bottommost gray hyperarc $f$ also satisfies $u \in t(f)$ and $v \in h(f)$, the three blue hyperarcs are heavier than it.
      Thus, the conditions of the $\lambda$-coreset with $\lambda = 3$ are satisfied for $(u, v)$.
      \Cref{subfig:critical} presents an image of critical pairs of three hyperarcs, which are missed by the coreset $S$ in \Cref{subfig:coreset}.
      Suppose that vertices $v$ have $x_v$ values of $2$, $3$, $4$, $5$, and $1$ from left to right, respectively, as shown nearby the vertices.
      Then, the green and yellow hyperarcs have $(u_1, v_1)$ and $(u_2, v_2)$, respectively, as $x$-ctirical pairs.
      If the three hyperarcs constitute $F_i^x \subseteq F \setminus S$, we have $E_i^x  = \set*{(u_1, v_1), (u_2, v_2)}$, and $F_i^x$ is partitioned into $F_i^{x, u_1v_1}$ and $F_i^{x, u_2v_2}$, shown in green and yellow, respectively.
  }
\end{figure}

Next, we explain {\DHOnestep}$(H,\lambda)$ given in \cref{alg:DHS-onestep}, which is the main subroutine in our algorithm.
The algorithm first computes a $\lambda$-coreset $S$ by calling {\CoresetFinder}$(H,\lambda)$. The hyperarcs in the coreset $S$ are deterministically added to the output.
Then, it randomly chooses the remaining hyperarcs with probability ${1}/{2}$ and doubles the weights if sampled, thus preserving the expected total weight.
The main technical observation is that, under an appropriate choice of $\lambda$, the output of {\DHOnestep}$(H,\lambda)$ is an $\eps$-spectral sparsifier of $H$.
Formally, we can show the following lemma, which is the main technical contribution and will be proved in \cref{section:DHS-JustificationForOneIteration}.
\begin{lemma}\label[lemma]{lemma:DHS-onestep}
  Let $H=(V,F,z)$ be a directed hypergraph with $|V|=n$ and $|F| = m$.
  For any $\eps\in (0,1)$ and $\lambda\geq \frac{\Ca \log^3 m}{\eps^2}$, where $\Ca$ is a sufficiently large constant, {\rm{\DHOnestep$(H,\lambda)$}} returns an $\eps$-spectral sparsifier $\Htl=(V,\tilde{F},\tilde{z})$ of $H$ satisfying
  $
  |\tilde{F}| \leq \frac{m}{2} + \prn*{3m\log n}^{\frac12} + \lambda n^2
  $
  with probability at least $1-O\left(\frac{1}{n^2}\right)$.
\end{lemma}

Finally, we present our sparsification algorithm {\DHSparsify}$(H,\eps)$ in \cref{alg:DHS-iterative}.
In the algorithm description, $\Ca$ denotes the constant given in the statement of \cref{lemma:DHS-onestep}, and $\Cc$ is a sufficiently large constant (which we can compute explicitly by carefully expanding the analysis in \cref{section:DHS-NearlyTightSparsification}).
The algorithm iteratively calls {\DHOnestep}$(\Htl_i,\lambda_i)$, where $\Htl_i$ is the sub-hypergraph obtained in the previous step.
Here, the parameter $\lambda_i$ is defined as in \cref{line:DHSparsify-eps-lambda}, which makes $\Htl_{i+1}$ an $\eps_i$-spectral sparsifier of $\Htl_i$ by the condition in \cref{lemma:DHS-onestep}.\footnote{Unlike the existing spanner-based algorithm \citep{koutis2016simple}, we need to change $\eps_i$ adaptively since fixing $\eps_{i}=\frac{\eps}{T}$ does not yield a sparsifier of the desired size when the input hypergraph is exponentially large in $n$.}
The algorithm repeatedly calls {\DHOnestep}$(\Htl_i,\lambda_i)$ until the size of $\Htl_i$ becomes $\Otl(n^2/\eps^2)$ or the maximum number of iterations, $T$, is reached.
With this choice of $\eps_i$, we will show that the size of $\Htl_i$ decreases geometrically and that the accumulated sparsification error is bounded by $\eps$.
Consequently, the final output is an $\eps$-spectral sparsifier of the desired size, which completes the proof of \cref{theorem:DHS-main}.
We present the analysis in \cref{section:DHS-NearlyTightSparsification}.

\subsection{Proof of \texorpdfstring{\cref{lemma:DHS-onestep}}{Lemma~\ref{lemma:DHS-onestep}}}\label{section:DHS-JustificationForOneIteration}

We prove \cref{lemma:DHS-onestep}, which ensures the correctness of {\DHOnestep}.
In this section, we let $H=(V,F,z)$, $\lambda\geq \frac{\Ca \log^3 m}{\eps^2}$, and $\eps\in(0,1)$ be as given in the statement of  \cref{lemma:DHS-onestep},
and let $\Htl=(V,\tilde{F},\tilde{z})$ be the output of \DHOnestep$(H,\lambda)$.

To prove \cref{lemma:DHS-onestep}, we bound the size and sparsification error of $\Htl$ from above.
The former is an easy consequence of the Chernoff bound.
We below prove it assuming $m> 12\log n$; otherwise, an input hypergraph is already sparsified and we do not run \DHOnestep.
\begin{lemma}\label[lemma]{lemma:DHS-size}
Let $H=(V,F,z)$ be a directed hypergraph with $|V| = n$ and $|F| = m$, and let $\lambda$ be a positive integer.
If $m> 12\log n$, \DHOnestep$(H,\lambda)$ outputs a sub-hypergraph $\Htl=(V, \Ftl, \ztl)$ of $H$ satisfying
$|\Ftl| \leq \frac{m}{2} + \prn*{3m\log n}^{\frac12} + \lambda n^2$ with probability at least $1-\frac{2}{n^2}$.
\end{lemma}
\begin{proof}
Let $S$ be a $\lambda$-coreset constructed in \cref{line:dhone-findcoreset} in \DHOnestep$(H,\lambda)$.
By \Cref{lemma:coreset}, $S$ has at most $\lambda n^2$ hyperarcs.
To bound $|\tilde{F}\setminus S|$, for each $f\in F\setminus S$, we let $X_f$ be a random variable that takes $1$ if $f$ is sampled and $0$ otherwise.
Note that $|\tilde{F}\setminus S|=\sum_{f\in F\setminus S} X_f$ holds.
Since we have $\mathbb{E}\left[\sum_{f\in F\setminus S}X_f \right] = (m-|S|)/2\leq m/2$, for any $t\in (0,1)$, the Chernoff bound (\cref{theorem:chernoff}) implies
\begin{align}
    \mathbb{P}\left[
        \sum_{f\in F\setminus S} X_f
        -
        \mathbb{E}\left[\sum_{f\in F\setminus S} X_f\right]
        > \frac{m}{2}t
    \right]
    \leq
    2\exp\left(
        -\frac{mt^2}{6}
    \right).
\end{align}
By setting $t=\left(\frac{12\log n}{m}\right)^{\frac12}$, which is smaller than $1$ due to the lemma assumption, we obtain
\begin{align}
    \mathbb{P}\left[
        \sum_{f\in F\setminus S} X_f
        \leq
        \frac{m}{2} + \left(3m\log n\right)^{\frac12}
    \right]
    \geq
    1 - \frac{2}{n^2}
    .
\end{align}
Thus, we have $|\tilde{F}| = |S| + \sum_{f\in F\setminus S} X_f \leq \frac{m}{2} + \left(3m\log n\right)^{\frac12} + \lambda n^2$ with probability at least $1-\frac{2}{n^2}$.
\end{proof}
%and hence we defer the proof to \cref{subsec:DHS-bounding-number-of-hyperarcs}.
The rest of this section focuses on showing that $\Htl$ is an $\eps$-spectral sparsifier of $H$, i.e., $(1-\eps) x^{\top} L_H(x)\leq x^{\top} L_{\Htl}(x) \leq (1+\eps) x^{\top} L_H(x)$ for any $x\in \R^V$.
Since this relation is invariant under scaling of $x$, it suffices to prove the relation for any $x$  satisfying $x^{\top} L_H(x)=1$.
Let $\mathbb{S}_H=\Set*{x\in \R^V}{x^{\top} L_H(x)=1}$.
A similar normalization is used in \citep{kapralov2021towards} with respect to the total energy of the corresponding underlying clique digraphs.
By contrast, we directly normalize the total energy of a hypergraph, $x^{\top} L_H(x)$.
This difference is a key to eliminating the extra $r^3$ factor, while it requires a new discretization scheme, as described later.

Since we analyze the contribution of each hyperarc to the energy of $H$, it is convenient to use the notation of $Q_H^x(f)$ and $Q_H^x(F')$ for $f\in F$ and $F'\subseteq F$, respectively, defined in \cref{section:preliminaries}.
Our goal is to prove $(1-\eps) Q_{H}^x(F)\leq Q_{\Htl}^x(\tilde{F}) \leq (1+\eps) Q_{H}^x(F)$ for all $x\in \mathbb{S}_H$.

Given $x\in\mathbb{S}_H$ and a $\lambda$-coreset $S \subseteq F$, our strategy is to partition $F \setminus S$ into subsets based on the energies and evaluate the error caused by sparsification for each subset.
Specifically, we classify hyperarcs $f \in F \setminus S$ into subsets $F^x_i$ defined for each $i \in \Z$ as follows:
\begin{align}
    F_i^x
    :=
    \Set*{f \in F\setminus S}{Q_H^x(f) \in \left[\frac{1}{2^i\lambda}, \frac{1}{2^{i-1}\lambda}\right)}.
\end{align}
We also define $\tilde{F}_i^x := F_i^x\cap \tilde{F}$ for each $i\in \Z$.

Since $Q_{H}^x(F \setminus S)=\sum_{i\in\mathbb{Z}} Q_{H}^x(F_i^x)$
and $Q_{\Htl}^x(\tilde{F} \setminus S)=\sum_{i\in\mathbb{Z}} Q_{\Htl}^x(\tilde{F}_i^x)$,
our goal is to prove that $|Q_{\Htl}^x(\tilde{F}_i^x)-Q_{H}^x(F_i^x)|$ is sufficiently small for all $i\in \mathbb{Z}$ and $x\in \mathbb{S}_H$.
This is not difficult if $i$ is sufficiently large, as in the following lemma.
\begin{restatable}{lemma}{iBound}\label[lemma]{lemma:DHS-onestep-iBound}
    Let $I = \lceil \log_2(9m) \rceil$.
    For any $x\in \mathbb{S}_H$, $\left|Q_{H}^x\left(\cup_{i\ge I+1} F_i^x \right)-Q^x_{\Htl}\left(\cup_{i\ge I+1} \tilde{F}_i^x \right)\right|\leq \frac{\eps}{3}$.
    % \begin{equation}\label{eq:DHS-onestep-ignore}
    % \left|Q_{H}^x\left(\cup_{i\ge I+1} F_i^x \right)-Q^x_{\Htl}\left(\cup_{i\ge I+1} \tilde{F}_i^x \right)\right|\leq \frac{\eps}{3}.
    % \end{equation}
\end{restatable}
\begin{proof}
    Due to the assumption in \cref{lemma:DHS-onestep}, $\lambda\eps\geq \frac{\Ca\log^3 m}{\eps}\geq 1$ holds for sufficiently large $\Ca$.
    By the definition of $F_i^x$, the energy of each hyperarc in $\cup_{i\ge I+1} F_i^x$ is less than $\frac{1}{2^I\lambda}$, which is at most $\frac{\eps}{9m}$ by $I = \lceil \log_2(9m) \rceil$ and $\lambda\eps\geq 1$.
    %Since $m\leq (n/\eps)^{r'}$,
    Thus, it holds that
        \begin{align}
            \label{eq:ignore1}
            Q_H^x\left(\cup_{i\ge I+1} F_i^x \right)
            =
            \sum_{f \in \cup_{i\ge I+1} F_i^x}Q_H^x (f)
            \leq m\cdot\frac{\epsilon}{9m}
            \leq
            \frac{\eps}{9}.
        \end{align}
    As for $\Ftl \subseteq F$, since the weight of each hyperarc in $\tilde{F}$ is doubled in \DHOnestep, we have
        \begin{align}
            \label{eq:ignore2}
            Q_{\Htl}^x\left(\cup_{i\ge I+1} \tilde{F}_i^x\right)
            \leq 2\cdot Q_H^x\left(\cup_{i\ge I+1} F_i^x \right)
            \leq
            \frac{2\eps}{9}
            .
        \end{align}
    Combining \cref{eq:ignore1,eq:ignore2}, we obtain the claim.
    %\cref{eq:DHS-onestep-ignore}.
\end{proof}

We then introduce additional definitions for the convenience of describing our discretization scheme and analyzing the sparsification error.
%To define our discretization scheme and perform the analysis, we need additional notation.
\begin{definition}
    For $x\in \mathbb{S}_H$, we say $(u, v) \in V\times V$ is an \textit{$x$-critical pair} of $f \in F$ if we have $(u,v)= \argmax_{(u,v)\in C(f)}(x_u-x_v)_+^2$, breaking ties as in \cref{section:preliminaries}.
    For $i\in \mathbb{Z}$ and $x\in \mathbb{S}_H$, let
    \begin{align}
        E_i^x = \Set*{(u,v) \in C(F)}{\text{$(u,v)$ is an $x$-critical pair of some $f \in F_i^x$}}
    \end{align}
    and, for each $(u,v)\in E_i^x$, let
    \begin{align}
        F_i^{x,uv} = \Set*{f \in F_i^x}{\text{$(u,v)$ is an $x$-critical pair of $f$}}.
    \end{align}
\end{definition}
Note that the collection of $F_i^{x,uv}$ for $(u, v) \in E_i^x$ forms a partition of $F_i^x$.
\Cref{subfig:critical} presents an example of $x$-critical pairs.

We now discuss how to bound $|Q_{\Htl}^x(\tilde{F}_i^x)-Q_{H}^x(F_i^x)|$ for $i$ that is not covered in \cref{lemma:DHS-onestep-iBound}.
By using the Chernoff bound, it is easy to evaluate the probability that $|Q_{\Htl}^x(\tilde{F}_i^x)-Q_{H}^x(F_i^x)|$ is small for each $x\in \mathbb{S}_H$.
To convert it to a uniform bound over all $x \in \mathbb{S}_H$, we construct an appropriate discretization scheme, as follows.

Let $\Delta=\frac{\eps}{9m}$.
For $(u,v)\in E_i^x$, we define the {\em discretization width} as $\Delta^{uv}_i \coloneqq \frac{\Delta}{\max_{f\in F_i^{x,uv}} z_f}$.
%$\Delta_i^{uv}$ as
%\[
%\Delta^{uv}_i=\frac{\Delta}{\max_{f\in F_i^{x,uv}} z_f}.
%\]
Note that $F_i^{x,uv}\neq \emptyset$ holds for $(u, v) \in E_i$
by the definitions of $E_i$ and $F_i^{x,uv}$, and hence $\Delta^{uv}_i$ is well-defined.
The denominator plays the role of scaling the width.
Given any $x\in \mathbb{S}_H$, we consider discretizing $(x_u-x_v)_+^2$ for each $(u,v)\in E_i^x$, not the energy itself.
Specifically, for each $i \in \mathbb{Z}$ and $(u,v)\in E_i^x$, we use $\left\lfloor \frac{(x_u-x_v)_+^2}{\Delta_i^{uv}}\right\rfloor\Delta_i^{uv}$ as a discretized value of $(x_u-x_v)_+^2$.
%\begin{align}
%    d_i^{uv}(x)=\left\lfloor \frac{(x_u-x_v)_+^2}{\Delta_i^{uv}}\right\rfloor\Delta_i^{uv}.
%\end{align}
Then, for each $f\in F_i^{x,uv}$ such that $(u,v)\in E_i^x$, we define the {\em discretized energy} $D^x_H(f)$ by
\[
D^x_H(f) \coloneqq z_f \left\lfloor \frac{(x_u-x_v)_+^2}{\Delta_i^{uv}}\right\rfloor\Delta_i^{uv}.
\]
It should be noted that the discretized energy of $f \in F_i^{x,uv}$ is defined by first discretizing $(x_u-x_v)_+^2$ and then scaling it by $z_f$.
This somewhat indirect discretization scheme will turn out important when bounding the number of possible discretized energies.
%(Note that $(x_u-x_v)_+^2$ depends on $(u,v)\in E_i^x$, while $D^x_H(f)$ depends on hyperarcs $f$. The discretization scheme enables us to count the number of discretized energies based on the size of $E_i^x$, which can be bounded sharply by using a $\lambda$ coreset in the output.)

For each sampled hyperarc $f\in F_i^{x,uv}\cap \Ftl$ with $(u,v)\in E_i^x$, we define the {\em discretized energy after sampling} by $D^x_{\Htl}(f) \coloneqq 2D^x_H(f)$.
We also let $D^x_H(F_i^x)=\sum_{f\in F_i^x} D_H^x(f)$ and $D^x_{\Htl}(\tilde{F}_i^x)=\sum_{f\in \tilde{F}_i^x} D_{\Htl}^x(f)$.
We can ensures that discretization errors are small as follows.
\begin{lemma}\label[lemma]{lemma:DHS-onestep-DError}
For any $x\in \mathbb{S}_H$, we have
\begin{align}
    \sum_{i\in \mathbb{Z}} \left| D^x_H(F_i^x) -  Q^x_H(F_i^x)\right| \leq \frac{\eps}{9}
    & &
    \text{and}
    & &
    \sum_{i\in \mathbb{Z}}\left |D^x_{\Htl}(\tilde{F}_i^x) - Q^x_{\Htl}(\tilde{F}_i^x)\right| \leq \frac{2\eps}{9}.
\end{align}
\end{lemma}
\begin{proof}
Recall that the discretized energy $D_H^x(f)$ of each $f\in F_i^{x,uv}$ is obtained by discretizing $(x_u-x_v)_+^2$ with the width $\Delta_i^{uv}$ and scaling it by $z_f$.
Therefore, the discretization error for each $f$ is bounded by $z_f\Delta_i^{uv}$.
From the definition of $\Delta_i^{uv}$, we have $z_f\Delta_i^{uv} = z_f \frac{\Delta}{\max_{f\in F_i^{x,uv}} z_f} \leq \Delta$.
Hence, the total discretization error over all $f \in F\setminus S$ is bounded by $m\Delta$, which is at most $\frac{\eps}{9}$ since $\Delta=\frac{\eps}{9m}$. Thus, we obtain the first inequality.
The second inequality follows from the fact that the weights of sampled hyperarcs are doubled.
\end{proof}

From \cref{lemma:DHS-onestep-DError}, we can bound the sparsification error  $|Q_{\Htl}^x(\tilde{F}_i^x)-Q_{H}^x(F_i^x)|$  for all $x\in \mathbb{S}_H$ by bounding $|D_{\Htl}^x(\tilde{F}_i^x)-D_{H}^x(F_i^x)|$ for all $x\in \mathbb{S}_H$.
Since the number of possible discretized energies is finite, we can use the standard Chernoff bound and union bound to evaluate the sparsification error.
Thus, what remains is to prove that the number of discretized energies is small enough so that we can obtain the desired uniform bound.
To this end, we first bound the size of $E_i^x$ and then bound the number of possible discretized values.
The following lemma bounds the size of $E_i^x$, in which the existence of a $\lambda$-coreset plays an important role.

\begin{lemma}\label[lemma]{lemma:DHS-onestep-MoreEnergy}
For $i\in \Z$, we have $\left|E_i^x\right| <  2^{i}$.
\end{lemma}
\begin{proof}
By the definition of $E_i^x$, for each $(u,v)\in E_i^x$, there is a hyperarc $f^{uv}\in F_i^x \subseteq F\setminus S$ such that
$(u,v)$ is an $x$-critical pair of $f^{uv}$.
Since $S$ is a $\lambda$-coreset, $S$ admits a partition $\Set*{S^{uv}}{(u,v)\in C(F)}$ satisfying the three conditions in \cref{lemma:coreset}.
Since $f^{uv}\notin S$, the third condition in \cref{lemma:coreset} implies $z_f\geq z_{f^{uv}}$ for any $f\in S^{uv}$.
Hence, for any $f\in S^{uv}$, we have
\begin{align}\label{eq:Ei}
Q_H^x(f^{uv}) = z_{f^{uv}} (x_u-x_v)_+^2 \leq z_{f} (x_u-x_v)_+^2 \leq \max_{(u',v')\in C(f)} z_{f} (x_{u'}-x_{v'})_+^2 = Q_H^x(f).
\end{align}
Since the second condition in \cref{lemma:coreset} implies $|S^{uv}| = \lambda$ for $(u,v) \in E_i^x \subseteq C(F\setminus S)$,
\begin{align}
        %1 & = Q_H^x(F)\qquad (\text{by } x\in \mathbb{S}_H)  \\
        Q_H^x(F) &\geq %\underset{\rm (a)}{\geq}
        \sum_{(u,v)\in E_i^x}\left( Q_H^x(S^{uv}) + Q_H^x(f^{uv})\right)\quad (\text{since all $S^{uv}$ and $f^{uv} \notin S$ are disjoint}) \\
        %& \geq \sum_{(u,v)\in E_i^x} (|S^{uv}|+1) \cdot Q_H^x(f^{uv})\qquad (\text{by \cref{eq:Ei}}) \\
        & \geq \sum_{(u,v)\in E_i^x} (\lambda+1) \cdot Q_H^x(f^{uv})\qquad (\text{by \cref{eq:Ei} and $|S^{uv}| = \lambda$}) \\
	& \geq \sum_{(u,v)\in E_i^x} (\lambda+1) \cdot (2^i\lambda)^{-1}\qquad  (\text{by } f^{uv}\in F_i^x).
	%&>2^{-i}|E_i^x|.
\end{align}
holds, hence $Q_H^x(F)>2^{-i}|E_i^x|$.
Since $Q_H^x(F) = 1$ by $x\in \mathbb{S}_H$, we obtain $|E_i^x|<2^i$.
\end{proof}
From \cref{lemma:DHS-onestep-MoreEnergy}, if $i \le 0$, we have $|E_i^x|< 2^i \leq 1$, which implies $E_i^x = \emptyset$ and $F_i^x = \emptyset$.
Thus, the following corollary holds.
%An easy but important corollary of \cref{lemma:DHS-onestep-MoreEnergy} is the following.
\begin{corollary}\label[corollary]{corollary:DHS-Fixisempty}
If $i\le 0$, we have $F_i^x=\emptyset$.
\end{corollary}
%\begin{proof}
%By \cref{lemma:DHS-onestep-MoreEnergy} and $i\le 0$, $|E_i^x|< 2^i \leq 1$,
%and hence $E_i^x = \emptyset$, implying $F_i^x = \emptyset$.
%\end{proof}

Due to \cref{corollary:DHS-Fixisempty} and \cref{lemma:DHS-onestep-iBound}, we can focus on $i \in \Z$ with $1 \le i \le I = \ceil*{\log_2 9m}$. In this range, we have the following bound on the number of possible discretized values.
\begin{lemma}\label[lemma]{lemma:number-of-values}
For each positive integer $i$, let $L_i = \Set*{\left(F_i^x, \{D^x_H(f)\}_{f\in F_i^x}\right)}{x\in \mathbb{S}_H}$, where $\{D^x_H(f)\}_{f\in F_i^x}$ is the list of the discretized energies over all hyperarcs in $F_i^x$.
If $1 \le i \le I = \ceil*{\log_2 9m}$, we have $|L_i|\leq \left( \frac{648 n^4m^4}{\lambda\eps} \right)^{2^i}$.
%\begin{align}
%        |L_i|\leq
%        \left( \frac{648 n^4m^4}{\lambda\eps} \right)^{2^i}.
%        \label{eq:DHS-onestep-counting}
%    \end{align}
\end{lemma}
Since the proof of  \cref{lemma:number-of-values} is not short, we first complete the proof of \cref{lemma:DHS-onestep} assuming that \cref{lemma:number-of-values} is true; then, we prove \cref{lemma:number-of-values} in \cref{subsubsec:proof-of-number-of-values-lemma}.

\begin{proof}[Proof of \cref{lemma:DHS-onestep}]
Let $I =  \lceil \log_2(9m) \rceil$ as in \cref{lemma:DHS-onestep-iBound} and define $L_i$ as in \cref{lemma:number-of-values}.
Fix $i \in \set*{1,2,\dots,I}$ and consider any element of $L_i$, which we denote by $(F_i^{y}, \{D_H^{y}(f)\}_{f\in F_i^{y}})$ for some $y\in \mathbb{S}_H$.
Since the discretized energy of each hyperarc is obtained by rounding down, we have $D^y_H(f)\leq Q^y_H(f)$.
Thus, for every $f\in F_i^y$, it holds that
\begin{align}\label{eq:onestep-1}
D_H^{y}(f) \leq Q^y_H(f) < \frac{1}{2^{i-1}\lambda}.
\end{align}
For each $f\in F\setminus S$, let $X_f$ be a random variable that takes $1$ with probability ${1}/{2}$ and $0$ otherwise, which represents the randomness of sampling and hence $D_{\Htl}^y (\tilde{F}_i^y) = \sum_{f\in F_i^y} 2X_fD_H^{y}(f)$.
By $D^y_H(f)\leq Q^y_H(f)$ again, we have
\begin{equation}\label{eq:onestep-2}
%\mathbb{E}\left[D_{\Htl}^y (\tilde{F}_i^y)\right]=
\mathbb{E}\Bigg[\sum_{f\in F_i^y} 2X_fD_H^{y}(f)\Bigg]
=
\sum_{f\in F_i^y}D_H^{y}(f)
=
D_H^{y}(F_i^y)
\leq Q^y_H(F_i^y)\leq Q^y_H(F)=1.
\end{equation}
Due to \cref{eq:onestep-1,eq:onestep-2}, the Chernoff bound (\cref{theorem:chernoff}) with $\mu =1$, $a = \frac{1}{2^{i-2}\lambda}$, and $\delta=\frac{\eps}{3I}$ implies
    \begin{align}
        \mathbb{P}\left[\left|
        D_{\Htl}^y (\tilde{F}_i^y)-D_H^{y}(F_i^y)
            \right| > \frac{\eps}{3I}
        \right]
        &=
        \mathbb{P}\left[\left|
        \sum_{f\in F_i^y} 2X_fD_H^{y}(f)
        -
        \mathbb{E}\left[\sum_{f\in F_i^y} 2X_fD_H^{y}(f)\right]
            \right| > \frac{\eps}{3I}
        \right]
        \\
        &\leq
        2\exp \left(
            - \frac{2^i \cdot \epsilon^2\lambda }{108I^2}
        \right).
    \end{align}
This bound is true for each $(F_i^{y}, \{D_H^{y}(f)\}_{f\in F_i^{y}})\in L_i$, and we can convert it to a uniform bound over all $(F_i^{y}, \{D_H^{y}(f)\}_{f\in F_i^{y}})\in L_i$ by using \cref{lemma:number-of-values} and the union bound as follows:
    \begin{align}
        \mathbb{P}\left[\exists (F_i^{y}, \{D_H^{y}(f)\}_{f\in F_i^{y}})\in L_i,  \left|
        D_{\Htl}^y (\tilde{F}_i^y)-D_H^{y}(F_i^y)
            \right| > \frac{\eps}{3I}
        \right]
        \leq
        2\exp \left(
            - \frac{2^i \cdot \epsilon^2\lambda }{108I^2}
        \right) \cdot \left(\frac{648 n^4m^4}{\lambda \eps}\right)^{2^{i}}.
    \end{align}
We may assume $nm\geq 648$ (otherwise \cref{lemma:DHS-onestep} is trivial for a sufficiently large $\Ca$).
Letting $\Ca$ be sufficiently large, we have $\lambda\geq \frac{\Ca \log^3 m}{\eps^2}\geq \frac{108 I^2}{\eps^2}(6\log n+5\log m)$ and $\lambda\eps \ge 1$.
%(and $\eps\lambda \geq 1$).
Thus, we can further bound the right-hand side from above by
\begin{align}
    2\exp \left(
                - \frac{2^i \cdot \epsilon^2\lambda }{108I^2}
            \right) \cdot \left( n^5m^5\right)^{2^{i}}
    \le
    2\exp \left(
                -2^i \cdot (6\log n+5\log m)
            \right) \cdot \left( n^5m^5\right)^{2^{i}}
    \le
    \frac{2}{n^{2^{i}}}.
\end{align}
Therefore, $\mathbb{P}[\forall (F_i^y, \{D_H^y(f)\}_{f\in F_i^y})\in L_i,  | D_{\Htl}^y(\tilde{F}_i^y)-D_H^y(F_i^y) | \leq \frac{\eps}{3I} ] \geq 1-\frac{2}{n^{2^{i}}}$ holds.
%\begin{align}
%        &\mathbb{P}\left[\forall (F_i^y, \{D_H^y(f)\}_{f\in F_i^y})\in L_i,  \left|
%        D_{\Htl}^y(\tilde{F}_i^y)-D_H^y(F_i^y)
%            \right| \leq \frac{\eps}{3I}
%        \right] \geq 1-\frac{2}{n^{2^{i}}}
%        .
%\end{align}
Since $(F_i^x, \{D^x(f)\}_{f\in F_i^x})\in L_i$ holds for all $x\in \mathbb{S}_H$, we can equivalently rewrite the bound as
\begin{align}
        &\mathbb{P}\left[\forall x\in \mathbb{S}_H,  \left|
        D^x_{\Htl}(\Ftl_i^x)-D^x_H(F_i^x)
            \right| \leq \frac{\eps}{3I}
        \right] \geq 1-\frac{2}{n^{2^{i}}}.
\end{align}
By the union bound over $1\leq i\leq I=\lceil \log_2 (9m)\rceil $ and $\sum_{i=1}^I\frac{2}{n^{2^{i}}}\leq \sum_{i=1}^\infty\frac{2}{n^{2i}}\leq \frac{2}{n^2-1}\leq \frac{3}{n^2}$ (for $n\geq 2$), we obtain
 \begin{align}\label{eq:onestep3}
        &\mathbb{P}\left[\forall x\in \mathbb{S}_H,
        \sum_{i=1}^I
        \left| D^x_{\Htl}(\Ftl_i^x)-D^x_H(F_i^x) \right|
        \leq \frac{\eps}{3}
        \right] \geq 1-\frac{3}{n^2}.
    \end{align}
Thus, for all $x\in \mathbb{S}_H$, we can bound
$|x^{\top} L_{\Htl}(x)-x^{\top} L_{H}(x)| = |Q_{\Htl}^x(\tilde{F})-Q_{H}^x(F)|$ as follows:
\begin{align}
&\left|Q_{\Htl}^x(\tilde{F})-Q_{H}^x(F)\right| \\
={}&\frac{\eps}{3}+\sum_{i=1}^{I} |Q_{\Htl}^x(\tilde{F}_i^x)-Q_{H}^x(F_i^x)| \qquad \prn*{ \text{by \Cref{lemma:DHS-onestep-iBound} and \Cref{corollary:DHS-Fixisempty}} } \\
\leq{}& \frac{\eps}{3}+ \sum_{i=1}^{I}\left[ |Q_{\Htl}^x(\tilde{F}_i^x)-D^x_{\Htl}(\tilde{F}_i^x)|+|D^x_{\Htl}(\tilde{F}_i^x)-D^x_H(F_i^x)|+|D^x_H(F_i^x)-Q_{H}^x(F_i^x)|\right] \\
\leq{}& \frac{\eps}{3}+\frac{\eps}{9}+\frac{\eps}{3}+\frac{2\eps}{9} \qquad \prn*{ \text{by \cref{lemma:DHS-onestep-DError} and \cref{eq:onestep3}} }\\
={}&\eps,
\end{align}
which holds with probability at least $1-\frac{3}{n^2}$.
Hence, $\Htl$ is an $\eps$-spectral sparsifier of $H$.
Combining this with the size bound in \cref{lemma:DHS-size}, we obtain \cref{lemma:DHS-onestep}.
\end{proof}

\subsubsection{Proof of \texorpdfstring{\cref{lemma:number-of-values}}{Lemma~\ref{lemma:number-of-values}}}\label{subsubsec:proof-of-number-of-values-lemma}
We present the proof of \cref{lemma:number-of-values}.
Our goal is to bound the size of $L_i$ defined in \cref{lemma:number-of-values} for $i \in \Z$ with $1\le i \le I = \ceil*{\log_2 9m}$.
To this end, we proceed in two steps:
we first bound the number of possible combinations of $(F_i^x, E_i^x, \{F_i^{x,uv}\}_{f\in E_i^x})$ over all $x\in \mathbb{S}_H$, and then bound the number of possible lists $\set*{D_H^x(f)}_{f \in F^x_i}$ of discretized energies.
For convenience, we define the following notion.
\begin{definition}\label[definition]{def:i_realized}
    Let $(E, \{f_{uv}\}_{(u,v)\in E}, \pi_E)$ be a tuple such that $E\subseteq V\times V$, $\{f_{uv}\}_{(u,v)\in E}$ is a list of hyperarcs indexed by $(u,v) \in E$, and $\pi_E$ is a total ordering on $E$.
    For $i \in \set*{1,2,\dots,I}$, we say $(E, \{f_{uv}\}_{(u,v)\in E}, \pi_E)$ is {\em $i$-realized} by $x\in \mathbb{S}_H$ if the following conditions hold:
    \begin{enumerate}
        \item $E = E_i^x$,
        \item $f_{uv}=\argmin_{f\in F_i^{x,uv}} z_f$ for each $(u,v)\in E_i^x$, and
        \item $\pi_{E}$ is the increasing order of the values of $(x_u-x_v)_+^2$, i.e., $(u,v)$ is smaller than $(u',v')$ in $\pi_E$ if and only if $(x_u-x_v)_+^2\leq (x_{u'}-x_{v'})_+^2$ (where the tie-breaking rule explained in \cref{section:preliminaries} is used when the equality holds).
    \end{enumerate}
\end{definition}

The following lemma says that the $i$-realizability determines $E_i^x, F_i^x$, and $F_i^{x,uv}$, implying that we can reduce the problem of counting the number of possible $(F_i^x, E_i^x, \{F_i^{x,uv}\}_{f\in E_i^x})$ to that of counting the number of possible tuples $(E, \{f_{uv}\}_{(u,v)\in E}, \pi_E)$.
\begin{lemma}\label[lemma]{lemma:counting1}
Let $(E, \{f_{uv}\}_{(u,v)\in E}, \pi_E)$ be a tuple as defined in \cref{def:i_realized} and $x, y \in \mathbb{S}_H$.
If both $x$ and $y$ $i$-realize $(E, \{f_{uv}\}_{(u,v)\in E}, \pi_E)$ and $\bigcup_{j=1}^{i-1} F_j^x=\bigcup_{j=1}^{i-1} F_j^y$ holds, then, for every $(u,v)\in E$, we have $E_i^x=E_i^y$, $F_i^x=F_i^y$, and $F_i^{x,uv}=F_i^{y,uv}$.
%\begin{align}
%    E_i^x=E_i^y, & & F_i^x=F_i^y, & & \text{and} & & F_i^{x,uv}=F_i^{y,uv}.
%\end{align}
\end{lemma}
\begin{proof}
By the definition of the $i$-realizability, we have $E_i^x=E=E_i^y$.
If we can assume $F_i^{x,uv}=F_i^{y,uv}$ for every $(u,v)\in E$, we have $F_i^x=\bigcup_{(u,v)\in E} F_i^{x,uv}=\bigcup_{(u,v)\in E} F_i^{y,uv}=F_i^y$ since $\Set*{F_i^{x,uv}}{(u,v)\in C(F)}$ and $\Set*{F_i^{y,uv}}{(u,v)\in C(F)}$ are partitions of $F_i^x$ and $F_i^y$, respectively.
Therefore, we below focus on proving $F_i^{x,uv}=F_i^{y,uv}$ for every $(u,v)\in E$.

For a contradiction, suppose $F_i^{x,u_1v_1}\neq F_i^{y,u_1v_1}$ for some $(u_1,v_1)\in E$.
Without loss of generality, we assume there is a hyperarc $f^* \in F_i^{x,u_1v_1}\setminus  F_i^{y,u_1v_1}$.
Since both $x$ and $y$ $i$-realize $(E, \{f_{uv}\}_{(u,v)\in E}, \pi_E)$ and $(u_1,v_1)\in E$,
the second condition of the $i$-realizability implies
\begin{equation}\label{eq:counting1-2}
\argmin_{f\in F_i^{x,u_1v_1}} z_f=f_{u_1v_1}=\argmin_{f\in F_i^{y,u_1v_1}} z_f.
\end{equation}
In particular, we have $z_{f_{u_1v_1}}\leq z_{f^*}$ for $f^*\in F_i^{x,u_1v_1}$.
Hence
\begin{align}
Q^y_H(f^*) &= z_{f^*}\max_{(u,v)\in C(f^*)}(y_{u}-y_{v})^2_+ \\
%&\geq z_{f^*} (y_{u_1}-y_{v_1})^2_+ \\
&\geq z_{f_{u_1v_1}}(y_{u_1}-y_{v_1})^2_+ \qquad \prn*{\text{by $(u_1, v_1) \in C(f^*)$ and $z_{f^*} \ge z_{f_{u_1v_1}}$}}\\
&= z_{f_{u_1v_1}}\max_{(u,v)\in C(f_{u_1v_1})} (y_u-y_v)^2_+ \qquad \prn*{\text{by $f_{u_1v_1}\in F_i^{y,u_1v_1}$ as in \cref{eq:counting1-2}}}\\
&\geq \frac{2^{-i}}{\lambda}\qquad (\text{by $f_{u_1v_1}\in F_i^{y}$})
.
\end{align}
From $Q^y_H(f^*)\geq \frac{2^{-i}}{\lambda}$ and $f^*\in F_i^{x,u_1v_1}\subseteq F\setminus S$, it must hold that $f^*\in \bigcup_{j=1}^{i} F_j^y$.
Moreover, since $\bigcup_{j=1}^{i-1} F_j^x=\bigcup_{j=1}^{i-1} F_j^y$ by the lemma assumption and $f^* \notin \bigcup_{j=1}^{i-1} F_j^x$ by $f^*\in F_i^{x,u_1v_1}$, we have $f^*\notin \bigcup_{j=1}^{i-1} F_j^y$, hence $f^* \in F_i^y$.
Since the orderings of $E$ with respect to $(x_u-x_v)_+^2$ and $(y_u-y_v)_+^2$ are both equal to $\pi_E$ and $(u_1, v_1)$ is an $x$-critical pair of $f^*$, we have
\begin{equation}\label{eq:counting1-1}
(u_1, v_1)=\argmax_{(u,v)\in C(f^*)\cap E} (y_u-y_v)_+^2.
\end{equation}
Since $f^*\in F_i^y$, \cref{eq:counting1-1} implies $f^*\in F_i^{y,u_1v_1}$, contradicting the assumption of $f^*\notin F_i^{y,u_1v_1}$.
Therefore, $F_i^{x,uv}=F_i^{y,uv}$ holds for every $(u,v)\in E$.
\end{proof}

\Cref{lemma:counting1} enables us to bound the number of possible $(F_i^x, E_i^x, \{F_i^{x,uv}\}_{f\in E_i^x})$ for $x\in \mathbb{S}_H$.
\begin{lemma}\label[lemma]{lemma:counting2}
For each $i\geq 1$, $\abs*{ \Set*{(F_i^x, E_i^x, \{F_i^{x,uv}\}_{f\in E_i^x}) }{ x\in \mathbb{S}_H } } \leq \left(2^i n^2 m\right)^{2^{i+1}}$ holds.
\end{lemma}
\begin{proof}
First, we suppose that $F_j^x$ for $j = 1,\dots,i-1$ are fixed.
Then, due to \cref{lemma:counting1}, we can bound the number of possible combinations of $(F_i^x, E_i^x, \{F_i^{x,uv}\}_{f\in E_i^x})$ for all $x\in \mathbb{S}_H$ by counting the number of possible tuples $(E, \{f_{uv}\}_{(u,v)\in E}, \pi_E)$ that can be $i$-realized by some $x\in \mathbb{S}_H$.
Since $|E|<2^i$ by \cref{lemma:DHS-onestep-MoreEnergy}, the number of possible $E$ is $\sum_{k=1}^{|E|}\binom{n^2}{k}\leq\sum_{k=1}^{2^i-1}\binom{n^2}{k}$.
Once $E$ is specified, there are up to $m$ possible choices of $f_{uv}$ for each $(u,v)\in E$.
Furthermore, the number of possible total orderings $\pi_E$ of $E$ is  at most $(|E|)!\leq (2^i)!$.
Thus, the number of possible tuples $(E, \{f_{uv}\}_{(u,v)\in E}, \pi_E)$ that can be $i$-realized by some $x\in \mathbb{S}_H$ is at most $\left(\sum_{k=1}^{2^i-1}\binom{n^2}{k}\right) \cdot m^{2^i} \cdot (2^i)!$.
This is further upper bounded by $\left(2^in^2m\right)^{2^i}$ by a simple calculation.
% \[
% \left(\sum_{k=1}^{|E|}{n^2\choose k}\right) \cdot m^{2^i} \cdot (2^i)!
% \leq 2^i\cdot (n^2)^{2^i}\cdot m^{2^i}\cdot (2^i)!
% \leq 2^i \left(2^in^2m\right)^{2^i}.
% \]

We now remove the assumption that $F_j^x$ for $j=1,\dots, i-1$ are fixed.
By inductively using the above bound in increasing order of $j$, the number of possible combinations of $(F_i^x, E_i^x, \{F_i^{x,uv}\}_{f\in E_i^x})$ over all $x\in \mathbb{S}_H$ is at most
$
\prod_{j=1}^i \left(2^j n^2m\right)^{2^j}
\leq
\left(2^i n^2 m\right)^{\sum_{j=1}^{i}2^j}
\leq
\left(2^i n^2 m\right)^{2^{i+1}}
$, thus completing the proof.
\end{proof}

We then fix any tuple $(F_i^y, E_i^y, \{F_i^{y,uv}\}_{f\in E_i^y})$ for some representative $y \in \mathbb{S}_H$ and upper bound the number of possible lists of discretized energies, $\{D^x_H(f)\}_{f\in F_i^x}$, over a subspace of $\mathbb{S}_H$ that consists of $x$ with $(F_i^x, E_i^x, \{F_i^{x,uv}\}_{f\in E_i^x}) = (F_i^y, E_i^y, \{F_i^{y,uv}\}_{f\in E_i^y})$.

% The above lemma gives an upper bound on the number of $i$-realizable tuples.
% Next, to bound $|L_i|$, we evaluate the number of possible lists of discretized energies, $\{D^x_H(f)\}_{f\in F_i^x}$, over all $x\in \mathbb{S}_H$, supposing that $x$ $i$-realizes some $i$-realizable tuple fixed arbitrarily.
\begin{lemma}\label[lemma]{lemma:counting3}
Let $i\geq 0$ and fix $y\in \mathbb{S}_H$ arbitrarily.
The number of possible lists $\{D^x_H(f)\}_{f\in F_i^x}$ for all $x\in \mathbb{S}_H$ with $(F_i^x, E_i^x, \{F_i^{x,uv}\}_{(u,v)\in E_i^x}) = (F_i^y, E_i^y, \{F_i^{y,uv}\}_{(u,v)\in E_i^y})$ is at most $\left( \frac{9m}{2^{i-2}\lambda \eps}\right)^{2^i}$.
\end{lemma}
\begin{proof}
Let $x \in \mathbb{S}_H$ satisfy the condition in the lemma statement and fix $(u,v)\in E_i^x$.
Since every $f\in F_i^{x,uv} \subseteq F_i^x$ satisfies $z_f(x_u-x_v)_+^2 = Q_H^x(f) < \frac{1}{2^{i-1}\lambda}$, the range of $(x_u-x_v)_+^2$ is restricted to $\Big[0, \frac{1}{2^{i-1}\lambda \min_{f\in F_i^{x,uv}} z_f}\Big)$.
Hence, the number of possible discretized $(x_u-x_v)_+^2$ values, $\left\lfloor {(x_u-x_v)_+^2}/{\Delta_i^{uv}}\right\rfloor\Delta_i^{uv}$, over all $x \in \mathbb{S}_H$ under the lemma condition is at most
\begin{equation}\label{eq:counting3-2}
    \frac{1}{\Delta_i^{uv}2^{i-1}\lambda \min_{ f\in F_i^{x,uv}} z_f}
    =
    \frac{1}{\Delta 2^{i-1}\lambda}\cdot \frac{\max_{ f\in F_i^{x,uv}} z_f}{\min_{ f\in F_i^{x,uv}} z_f}
    \le
    \frac{1}{\Delta 2^{i-2}\lambda},
\end{equation}
where the equality is due to $\Delta^{uv}_i={\Delta}/{\max_{f\in F_i^{x,uv}} z_f}$ and the inequality comes from
$z_f(x_u-x_v)_+^2 = Q_H^x(f) \in \left[ \frac{1}{2^i\lambda} , \frac{1}{2^{i-1}\lambda} \right)$ for $f\in F_i^{x,uv} \subseteq F_i^x$, i.e., $\max_{ f\in F_i^{x,uv}} z_f \leq 2\min_{ f\in F_i^{x,uv}} z_f$.

Since the discretized energy of $f\in F_i^{x,uv}$ is defined by $D^x_H(f) = z_f \left\lfloor {(x_u-x_v)_+^2}/{\Delta_i^{uv}}\right\rfloor\Delta_i^{uv}$, fixing the discretization of $(x_u-x_v)_+^2$ determines discretized energies of all $f\in F_i^{x,uv}$.
Therefore, the number of possible lists $\{D^x_H(f)\}_{f\in F_i^{x,uv}}$ is also bounded by \cref{eq:counting3-2} for each $(u,v)\in E_i^x$.
Since $|E_i^x| < 2^i$ by \cref{lemma:DHS-onestep-MoreEnergy}, the number of possible lists $\{D^x_H(f)\}_{f\in F_i^x}$ is at most $\big(\frac{1}{\Delta 2^{i-2}\lambda}\big)^{2^i}$.
By substituting $\Delta=\frac{\eps}{9m}$ into it, we obtain the lemma.
\end{proof}

We are now ready to prove \cref{lemma:number-of-values}.
\begin{proof}[Proof of \cref{lemma:number-of-values}]
We can uniquely specify any element of $L_i$ by first fixing $(F_i^x, E_i^x, \{F_i^{x,uv}\}_{f\in E_i^x})$ and then $\{D^x_H(f)\}_{f\in F_i^x}$.
Therefore, we have
$
|L_i|\leq (2^i n^2 m)^{2^{i+1}}\cdot \left( \frac{9m}{2^{i-2}\lambda \eps}\right)^{2^i}
=
\left(
\frac{36\cdot 2^in^4m^3}{\lambda\eps}
\right)^{2^i}
$ by \cref{lemma:counting2,lemma:counting3}.
Combining this with $i \le I =  \lceil \log_2 9m \rceil$ completes the proof.
\end{proof}

\subsection{Proof of \texorpdfstring{\cref{theorem:DHS-main}}{Theorem~\ref{theorem:DHS-main}}}\label{section:DHS-NearlyTightSparsification}
Let $H=(V,F,z)$ be a directed hypergraph with $|V|=n$ and $|F|=m$, $\eps\in (0,1)$, and $\Htl=(V,\tilde{F},\tilde{z})$ the output of {\DHSparsify}$(H,\eps)$.
 Our goal is to prove that $\Htl$ is an $\eps$-spectral sparsifier of $H$ and $|\tilde{F}|=O\left(\frac{n^2}{\eps^2}\log^3\frac{n}{\eps}\right)$.
We here use $m^*$, $T$, $i_{\rm end}$, $(\Htl_i = (V,\tilde{F}_i, \tilde{z_i}), \lambda_i)$, $m_i$, and $\eps_i$ given in the description of {\DHSparsify}$(H,\eps)$ (\cref{alg:DHS-iterative}), where
$m^*=\frac{n^2}{\eps^2}\log^3 \frac{n}{\eps}$ is the target sparsifier size,
$T = \ceil*{\log_{4/3} \left(\frac{m}{m^*}\right)}$ is the maximum number of iterations,
$\Tend$ is the number of iterations performed,
($\Htl_i = (V,\tilde{F}_i, \tilde{z_i})$, $\lambda_i$) is the input of {\DHOnestep} at the $i$th iteration, $m_i=|\tilde{F}_i|$, and $\eps_i = \frac{\eps}{4 \log_{4/3}^2\left(\frac{m_i}{m^*}\right)}$, as in \cref{line:DHSparsify-eps-lambda} of \cref{alg:DHS-iterative}.

We first show that the number of hyperarcs decreases geometrically in each step.
\begin{lemma}\label[lemma]{lemma:DH-threequarter}
    Let $m_i$ be the number of hyperarcs in $\Htl_i$.
    Assume $m_i \ge \Cc m^* = \Cc \frac{n^2}{\eps^2} \log^3 \frac{n}{\eps}$ for a sufficiently large constant $\Cc$. Then, we have $\prn*{3m_i\log n}^{\frac12} + \lambda_i n^2 \leq \frac{m_i}{4}$.
\end{lemma}
\begin{proof}
It is easy to show that $\prn*{3m_i\log n}^{\frac12} \leq \frac{m_i}{8}$ holds if $m_i\geq 192\log n$, which is true if $\Cc$ is sufficiently large.
Hence, the desired inequality holds if $\lambda_i n^2 \leq \frac{m_i}{8}$, which we show below.

By \cref{line:DHSparsify-eps-lambda} in \cref{alg:DHS-iterative}, we have $\eps_i=\frac{\eps}{4\log^2_{4/3}\frac{m_i}{m^*}}$ and
$\lambda_i=\ceil*{ \frac{\Ca \log^3 m_i}{\eps_i^2} }$.
Hence,
\[
    \frac{m_i}{8}-\lambda_i n^2\geq \frac{m_i}{8}-\frac{2500\Ca n^2}{\eps^2}\log^3 m_i \log^4\frac{m_i}{m^*}
    \qquad
    \text{(by $4^2/\log^4(4/3) < 2500$)}.
\]
Let $m_i=\alpha m_*$ and $g(\alpha)$ be the right-hand side of the above inequality, which we regard as a function of $\alpha$.
Since $m^*=(n/\eps)^2\log^3({n}/{\eps})$, we have
\begin{align}
    g(\alpha) &= m_*
    \left(
        \frac{\alpha}{8} - \frac{2500\Ca}{\log^3(n/\eps)}\log^3 (\alpha m_*)\log^4\alpha
    \right)
    \\
    &\geq
    m_*\left(
        \frac{\alpha}{8} - \frac{10000\Ca}{\log^3(n/\eps)}(\log^3 \alpha + \log^3 m_*) \log^4\alpha
    \right)\quad (\text{by $(a+b)^3 \leq 4(a^3 + b^3)$})
    \\
    &\geq
    m_*\left(
        \frac{\alpha}{8} - 10000\Ca\left(\frac{\log^3 \alpha}{\log^3(n/\eps)} + 125\right) \log^4\alpha
    \right)\quad \Big(\text{by $m_* = \frac{n^2}{\eps^2}\log^3 \frac{n}{\eps} \leq \left(\frac{n}{\eps}\right)^5$}\Big).
\end{align}
Thus, there exists a sufficiently large constant $\Cc$, which is independent of $n$ and $\eps$, such that $g(\alpha)\geq 0$ holds for all $\alpha \ge \Cc$.
Using this constant $\Cc$, for all $m_i\geq \Cc m^*$, we have $\lambda_i n^2\leq \frac{m_i}{8}$ as desired.
\end{proof}

\begin{proof}[Proof of \cref{theorem:DHS-main}]
We say \DHOnestep$(H_i,\lambda_i)$ is {\em successful} if $\Htl_{i+1}$ is an $\eps_{i}$-spectral sparsifier of $\Htl_i$ and $m_{i+1} \leq \frac34 m_i$ holds.
{\DHSparsify}$(H,\eps)$ calls \DHOnestep$(H_i,\lambda_i)$ only when $m_i \geq \Cc m^*$ and $i\leq T$.
Therefore, by \cref{lemma:DHS-onestep,lemma:DH-threequarter}, with probability at least $1-O\left(\frac{T}{n^2}\right) \gtrsim 1 - O\left(\frac{1}{n}\right)$, \DHOnestep$(H_i,\lambda_i)$ is successful for all $i$ with $0\leq i\leq i_{\rm end}$.
Hence, assuming all \DHOnestep$(H_i,\lambda_i)$ to be successful, we below prove that the output hypergraph $\Htl$ has $O\left(\frac{n^2}{\eps^2}\log^3\frac{n}{\eps}\right)$ hyperarcs and that $\Htl$ is an $\eps$-spectral sparsifier of $H$.

We first discuss the size of $\Htl$.
If $m_i \leq \Cc m^*=\frac{\Cc n^2\log^3 (n/\eps)}{\eps^2}$ occurs for some $i\leq T-1$, then $m_i$ gives the size of $\Htl$ by the termination rule of {\DHSparsify}, which is already small enough.
Hence we below assume $m_i \geq \Cc m^*$ for all $i< T$.
Since every \DHOnestep$(H_i,\lambda_i)$ is successful,
$m_{i+1} \leq \frac34 m_i$ holds for all $i=0,1,\cdots,T-1$.
Thus, it holds that
\begin{align}
    m_{T} \leq m \cdot \left(\frac34\right)^T
    \leq
    m \cdot \left(\frac34\right)^{\log_{4/3}\frac{m\eps^2}{n^2\log^3 (n/\eps)}}
    =
    \frac{n^2\log^3 (n/\eps)}{\eps^2}.
\end{align}
Therefore, we have $|\Ftl| = O\left(\frac{n^2}{\eps^2}\log^3\frac{n}{\eps}\right)$.

We then show that $\Htl$ is an $\eps$-spectral sparsifier of $H$.
Since $\Htl_{i+1}$ is an $\eps_i$-spectral sparsifier of $\Htl_i$ for all $i=0,1,\cdots, \Tend-1$, the output hypergraph $\Htl = \Htl_{\Tend}$ is an $\tilde{\eps}$-spectral sparsifier of $H$, where
\begin{align}
    \tilde{\eps} = \max\left\{
    \prod_{i=0}^{\Tend-1} \left(1 + \eps_i\right) - 1,
    1 - \prod_{i=0}^{\Tend-1} \left(1 - \eps_i\right)
    \right\}
    .
\end{align}
A simple calculation yields the following upper bound on $\tilde{\eps}$:
\begin{align}\label{eq:accumulatederrors-1}
    \tilde{\eps}
    \leq
    \sum_{j=1}^{\Tend}
    \sum_{0\leq i_1<\cdots<i_j\leq \Tend-1}
    \eps_{i_1}\eps_{i_2}\cdots\eps_{i_j}
    \leq
    \sum_{j=1}^{\Tend}
    %\frac{1}{j!}
    \left(\sum_{i=0}^{\Tend-1} \eps_i\right)^j
    % \leq
    % \exp \left(\sum_{i=0}^{\Tend} \eps_i\right) - 1
    .
\end{align}
Since $m_{i+1}\leq \frac34 m_i$ and $m_{\Tend - 1}\geq \Cc m^*$,
we have $m_{\Tend -j}\geq \left(\frac{4}{3}\right)^{j-1} \Cc m^* \ge \prn*{\frac43}^j m^*$ for sufficiently large $\Cc \ge \frac43$, hence $\log_{4/3} \left(\frac{m_{\Tend -j}}{m^*}\right)\geq j$.
Using $\sum_{j =1}^\infty \frac{1}{j^2} \le \frac{\pi^2}{6}$,
we obtain
\begin{align}\label{eq:accumulatederrors-2}
\sum_{i=0}^{\Tend-1} \eps_i = \sum_{i=0}^{\Tend-1} \frac{\eps}{4 \log_{4/3}^2\left(\frac{m_i}{m^*}\right)}
\leq
\sum_{j=1}^{\infty} \frac{\eps}{4 j^2}
\leq
\frac{\eps}{4} \cdot \frac{\pi^2}{6}\leq \frac{\eps}{2}
.
\end{align}
Putting this into the right-hand side of \cref{eq:accumulatederrors-1}, we have
\begin{align}\label{eq:accumulatederrors-3}
    \sum_{j=1}^{\Tend}
    \left(\sum_{i=0}^{\Tend-1} \eps_i\right)^j
    \leq
    \sum_{j=1}^{\Tend} \prn*{\frac{\eps}{2}}^j
    \leq
    \frac{\frac{\eps}{2}}{1-\frac{\eps}{2}}
    \leq
    \eps.
\end{align}
%holds for any $\eps \in (0, 1)$.
By \cref{eq:accumulatederrors-1,eq:accumulatederrors-3}, $\Htl=\Htl_{\Tend}$ is an $\eps$-spectral sparsifier of $H$.

To conclude, with probability at least $1-O\left(\frac{1}{n}\right)$, {\rm DH-Sparsify}$(H,\eps)$ outputs an $\eps$-spectral sparsifier of $H$ with $O\left(\frac{n^2}{\eps^2}\log^3\frac{n}{\eps}\right)$ hyperarcs.
\end{proof}

\subsection{Total Time Complexity}\label{section:DHS-ComputationalComplexity}
We show that our algorithm runs in $O(r^2 m)$ time with probability at least $1 - O(1/n)$.
\begin{theorem}\label[theorem]{theorem:DHS-complexity}
For any directed hypergraph $H=(V,F,z)$ with the rank $r$ and $m$ hyperarcs and $\eps \in (0, 1)$, {\DHSparsify}$(H, \eps)$ runs in $O(r^2m)$ time with probability at least $1 - O(1/n)$.
\end{theorem}
\begin{proof}
We first discuss the running time of \DHOnestep$(\Htl_i, \lambda_i)$, where $\Htl_i = (V, \Ftl_i, \ztl_i)$ and $|\Ftl_i| = m_i$.
It first constructs a $\lambda_i$-coreset by calling \CoresetFinder$(\Htl_i, \lambda_i)$.
{\CoresetFinder} first constructs $A^{uv} = \Set*{f \in F}{C(f) \ni (u,v)}$ for $(u,v) \in C(F)$, which is done in $O(r^2 m_i)$ time since we have $|C(f)| = O(r^2)$ for each $f \in \Ftl_i$.
Then, for each $(u,v) \in C(F)$, it selects the $\lambda_i$ heaviest hyperarcs from $A^{uv}\setminus S$ in $O(|A^{uv}\setminus S|)$ time by using a selection algorithm \citep{blum1973time}, thus taking $O\prn*{\sum_{(u, v) \in C(F)}|A^{uv}\setminus S|} = O(r^2 m_i)$ time in total.
Therefore, \CoresetFinder$(\Htl_i, \lambda_i)$ takes $O(r^2m_i)$ time.
After that, {\DHOnestep} samples the remaining hyperarcs in $O(m_i)$ time.
Thus, \DHOnestep$(\Htl_i, \lambda_i)$ takes $O(r^2m_i)$ time.

We then bound the total time complexity.
Since {\DHSparsify}$(H, \eps)$ calls \DHOnestep$(\Htl_i, \lambda_i)$ for $i = 0,1,\dots,T-1$ (or stops earlier), the total time complexity is at most $O\prn*{r^2 \sum_{i=0}^{T-1} m_i}$.
From \cref{lemma:DHS-size,lemma:DH-threequarter}, whenever {\DHOnestep} is called, we have $m_{i+1} \le \frac34 m_i$ with probability at least $1 - O(1/n^2)$.
This implies that $\sum_{i=0}^{T-1} m_i \le m \sum_{i=0}^{T-1} \prn*{\frac34}^{i} \le 4m$ holds with probability at least $1 - O(T/n^2) \gtrsim 1 - O(1/n)$.
Therefore, the total time complexity is bounded by $O(r^2 m)$ with probability at least $1 - O(1/n)$.
\end{proof}

\section{Lower Bound on Sparsification of Directed Hypergraphs}\label{section:LowerBound}
We show the existence of an unsparsifiable directed hypergraph with $\Omega(n^2/\eps)$ hyperarcs and the rank three, thus proving \cref{theorem:Intro-DH-lower}.

\begin{proof}[Proof of \cref{theorem:Intro-DH-lower}]
    We construct an unsparsifiable bipartite directed hypergraph $H = (U\cup W, F, z)$ with $2n$ vertices such that $|U| = |W| = n$.
    For simplicity, we assume that $\frac{1}{8\eps}$ is an integer.
    We label vertices in $U$ and $W$ as $U = \set*{1,2,\dots,n}$ and $W = \set*{n+1,n+2\dots,2n}$, respectively.
    For each pair of $i \in U$ and $k \in W$
    and an integer $j$ with $j = i+1,i+2,\dots,i+\frac{1}{8\eps} \pmod{n}$, we create a hyperarc $f_{i,j,k}$ such that
    \begin{align}
    t(f_{i,j,k}) = \set*{i, j} \subseteq U && \text{and} &&
    h(f_{i,j,k}) = \set*{k} \in W
    .
    \end{align}
    These $\Omega(\frac{n^2}{\eps})$ hyperarcs form the hyperarc set, $F$, of $H$.
    Note that no multiple hyperarcs are defined since we have $\frac{1}{4\eps} < n$, as assumed in the statement of \cref{theorem:Intro-DH-lower}.
    We let $z_f = 4\eps$ for every $f \in F$.
    Note that $H$ is defined in a symmetric manner so that each pair $(i, k) \in U \times W$ is contained in $\frac{1}{4\eps}$ hyperarcs.

    Suppose for a contradiction that a proper sub-hypergraph $\Htl = (V, \Ftl, \ztl)$ of $H$ is an $\eps$-spectral sparsifier of $H$.
    Without loss of generality, we suppose that $f_{1,\js,n+1}\in F \setminus \Ftl$ holds for some $s$ with $2 \le \js \le 1 + \frac{1}{8\eps}$.

    We define $x^1 \in \R^{U\cup W}$ as follows:
    for $i \in U$, $x^1_i = 1$ if $i = 1$ and $0$ otherwise;
    for $i \in W$, $x^1_i = 0$ if $i = n+1$ and $1$ otherwise.
    We also define $x^s \in \R^{U\cup W}$ as a vector obtained by swapping the first and the $s$th elements of $x^1$.
    That is, $x^1$ and $x^s$ are written as follows:
    \begin{align}
        x^1 = (\overbrace{\underset{\substack{\uparrow \\ \text{$1$st}}}{1}, 0, \dots, 0}^{U},
             \overbrace{\underset{\substack{\uparrow \\ \text{ \makebox[0pt]{($n+1$)st}  }}}{0}, 1, \dots, 1}^W)^\top
        & &
        \text{and}
        & &
        x^s = (\overbrace{0, \dots, 0, \underset{\substack{\uparrow \\ \text{$s$th}}}{1}, 0, \dots, 0}^{U},
             \overbrace{\underset{\substack{\uparrow \\ \text{ \makebox[0pt]{($n+1$)st}  }}}{0}, 1, \dots, 1}^W)^\top.
    \end{align}
    We also define $x^{1s} = x^1 \vee x^s \in \R^{U\cup W}$, where $\vee$ denotes the element-wise maximum operator.

    Since $x^1\in \{0,1\}^V$, ${x^1}^\top L_H(x^1)$ is the value of some cut in $H$.
    Specifically,
    since  $(x^1_i - x^1_j)_+^2$ is $1$ if $(i, j) = (1, n+1)$ and $0$ otherwise,
    ${x^1}^\top L_H(x^1)$ is the value of the cut-set $F_1 = \left\{f \in F\ |\ \text{$1 \in t(f)$ and}\right.$ $\left.(n+1) \in h(f)\right\}$.
    Since $F_1$ has $\frac{1}{4\eps}$ hyperarcs, it holds that
    \[
        {x^1}^\top L_H(x^1) = \sum_{f \in F} 4\eps \cdot \max_{(i, j) \in C(f)} (x^1_i - x^1_j)_+^2 = 4\eps \cdot|F_1| = 1.
    \]
    Similarly, we have ${x^s}^\top L_H(x^s) = 4\eps \cdot |F_s| = 1$, where $F_s = \Set*{f \in F}{\text{$s \in t(f)$ and $(n+1) \in h(f)$}}$.
    From the assumption that $\Htl$ is an $\eps$-spectral sparsifier of $H$, we have
    \begin{align}\label{eq:lower_bound_1}
        \begin{aligned}
        &
        1 - \eps = (1 - \eps) {x^1}^\top L_H(x^1) \le {x^1}^\top L_{\Htl}(x^1)
        \quad \text{and} \\
        &
        1 - \eps = (1 - \eps) {x^s}^\top L_H(x^s) \le {x^s}^\top L_{\Htl}(x^s).
        \end{aligned}
    \end{align}

    Again, for the same reason as above,  ${x^{1s}}^\top L_H(x^{1s})$ is the value of the cut-set $F_1 \cup F_s$ in $H$.
    Since $F_1 \cap F_s = \set*{f_{1,s,n+1}}$, we have $|F_1 \cup F_s| = 2\cdot\frac{1}{4\eps} - 1$, and hence
    \[
        {x^{1s}}^\top L_{H}(x^{1s}) = 4\eps\cdot \prn*{ 2\cdot \frac{1}{4\eps} - 1} = 2(1 - 2\eps).
    \]
    Combining this with ${x^{1s}}^\top L_{\Htl}(x^{1s}) \le (1 + \eps){x^{1s}}^\top L_H(x^{1s})$, we obtain
    \begin{align}\label{eq:lower_bound_2}
        {x^{1s}}^\top L_{\Htl}(x^{1s}) \le (1 + \eps) \cdot 2(1 - 2\eps) = 2(1 - \eps - 2\eps^2).
    \end{align}

    On the other hand, ${x^{1s}}^\top L_{\Htl}(x^{1s})$ equals the value of the cut-set $\set*{F_1 \cup F_s} \cap \Ftl$ in  $\Htl$.
    Since we have $f_{1,s,n+1} \notin \Ftl$, $F_1 \cap \Ftl$ and $F_s \cap \Ftl$ are disjoint, which implies ${x^1}^\top L_{\Htl}(x^1) + {x^s}^\top L_{\Htl}(x^s) = {x^{1s}}^\top L_{\Htl}(x^{1s})$.
    Therefore, from \cref{eq:lower_bound_1,eq:lower_bound_2}, we have
    \[
        2(1 - \eps) \le {x^1}^\top L_{\Htl}(x^1) + {x^s}^\top L_{\Htl}(x^s) = {x^{1s}}^\top L_{\Htl}(x^{1s}) \le 2(1 - \eps - 2\eps^2),
    \]
    yielding a contradiction.
\end{proof}

\section{Application to Agnostic Learning of Submodular Functions}\label{sec:submodular}
We apply our \cref{theorem:DHS-main} to agnostic learning of submodular functions based on the discussion by Soma~and~Yoshida~\citep{soma2019spectral}.
Our goal is to prove the following sample complexity bound.
\begin{theorem}\label[theorem]{theorem:submodular}
    Let $V$ be a finite set with $n=|V|$,
    $\mathcal{C}$ be the class of nonnegative hypernetwork-type submodular functions $\sfn\colon 2^V\to [0,1]$, and $\mathcal{D}$ be any distribution on $2^V\times [0,1]$.
    For any $\eps, \delta \in (0, 1)$, there exists an algorithm that, given $\tilde{O}\prn*{ \frac{n^3}{\eps^4} + \frac{1}{\eps^2}\log\frac{1}{\delta} }$ i.i.d. samples from $\mathcal{D}$, outputs a submodular function $\sfn'\colon 2^V \to [0,1]$ that can be evaluated in polynomial time in $n$ and satisfies
    \begin{align}
        \mathbb{E}_{(X,y)\sim \mathcal{D}} \left[|\sfn'(X)-y|\right]\leq \inf_{\sfn\in \mathcal{C}} \mathbb{E}_{(X,y)\sim \mathcal{D}} \left[|\sfn(X)-y|\right] + \eps
    \end{align}
    with probability at least $1-\delta$.
\end{theorem}
This result improves previous bounds of  $n^{O(1/\eps^2)}\log\frac{1}{\delta}$ in~\citep{cheraghchi2012submodular} and  %$\tilde{O}\left(\frac{n^4}{\eps^4} + \frac{\log (1/\delta)}{\eps^2}\right)$
$\tilde{O}\prn*{ \frac{n^4}{\eps^4} + \frac{1}{\eps^2}\log\frac{1}{\delta} }$ in \citep{soma2019spectral}.
We remark that, as is already mentioned in \citep{soma2019spectral}, the setting of \citep{cheraghchi2012submodular} restricts the marginal distribution of $\mathcal{D}$ over $2^V$ to a product distribution and allows the output function to be non-submodular.
On the other hand, the algorithm of~\citep{cheraghchi2012submodular} takes only $n^{O(1/\eps^2)}\log\frac{1}{\delta}$ time, while both that of \citep{soma2019spectral} and ours take exponential time in $n$ in general.
For more information on this topic, we refer the reader to \citep{balcan2018learning,soma2019spectral}.

Given the discussion in~\citep{soma2019spectral}, the application is mostly straightforward.
We, however, need to care about hyperarc weights resulting from our algorithm because of its iterative nature.
(If a sparsifier has hyperarcs with exponentially heavy weights, the following discussion ceases to hold.)
As shown below, we can avoid this issue by first using the result of \citep{soma2019spectral} and then our method to obtain a sparsifier with polynomially heavy hyperarcs, thus obtaining the improved sample complexity bound as in \cref{theorem:submodular}.

Let us now introduce the background.
We say a set function $\sfn\colon2^V \to \R$ is \textit{submodular} if $\sfn(X) + \sfn(Y) \ge \sfn(X \cup Y) + \sfn(X \cap Y)$ holds for every $X, Y \subseteq V$.
One large subclass of submodular functions is the class of nonnegative hypernetwork-type functions $\sfn\colon2^V \to \R$, which are defined by the following conditions:
$\sfn(\emptyset) = \sfn(V) = 0$, $\sfn(X) \ge 0$ for any $X \subseteq V$, and $\sum_{Y \subseteq X} (-1)^{|Y\setminus X|} \sfn(Y) \le 0$ for any $X \subseteq V$ with $|X| = 2,3,\dots,|V|$.
Fujishige~and~Patkar~\citep{fujishige2001realization} showed that any nonnegative hypernetwork-type submodular function can be represented as a cut function of some directed hypergraph on $|V|$ vertices.
Due to the analysis in \citep{fujishige2001realization} and the existence of $\eps$-spectral sparsifiers of size $O\left(\frac{n^3}{\eps^2} \log n\right)$ by Soma~and~Yoshida~\citep{soma2019spectral}, the following claim holds.
\begin{proposition}[{\citep{soma2019spectral}}]\label[proposition]{proposition:sy-hypernetwork-approx} % Corollary 1.1
    Let $V$ be a finite set with $|V| = n$.
    For any nonnegative hypernetwork-type submodular function $\sfn\colon2^V \to \R$ and any $\eps \in (0, 1)$, there is a directed hypergraph $H = (V, F, z)$ with $O\left(\frac{n^3}{\eps^2} \log n\right)$ hyperarcs such that the cut function of $H$, denoted by $\kappa_H:2^V \to \R$, satisfies $(1-\eps)\kappa_H(X) \le \sfn(X) \le (1+\eps)\kappa_H(X)$ for every $X \subseteq V$.
    Moreover, $z_f \le n^3$ holds for every $f \in F$.
\end{proposition}

Combining \cref{proposition:sy-hypernetwork-approx} with \cref{theorem:DHS-main}, we can show that any nonnegative hypernetwork-type submodular function can be approximated by a cut function of some directed hypergraph with fewer hyperarcs as follows.

\begin{lemma}\label[lemma]{lemma:ours-hypernetwork-approx}
    Let $V$ be a finite set with $|V| = n$.
    For any nonnegative hypernetwork-type submodular function $\sfn\colon2^V \to [0,1]$ and $\eps \in (0, 1)$, there is a directed hypergraph $\Htl = (V, \Ftl, \ztl)$ with $O\left(\frac{n^2}{\eps^2} \log^3\frac{n}{\eps}\right)$ hyperarcs such that $(1-\eps)\kappa_{\Htl}(X) \le \sfn(X) \le (1+\eps)\kappa_{\Htl}(X)$ and $0\leq \kappa_{\Htl}(X)\leq 1$ hold for every $X \subseteq V$.
    Moreover, $\ztl_f =O(n^{5.5})$ holds for every $f \in \Ftl$.
\end{lemma}

\begin{proof}
    Due to \cref{proposition:sy-hypernetwork-approx}, there is a directed hypergraph $H = (V,F,z)$ that satisfies $|F| = O\left(\frac{n^3}{\eps^2} \log n\right)$, $(1-\eps/6)\kappa_H(X) \le \sfn(X) \le (1+\eps/6)\kappa_H(X)$ for every $X \subseteq V$, and $z_f \le n^3$ for every $f \in F$.
    By applying {\DHSparsify}$(H, \eps/6)$ to this $H$, we can obtain an ($\eps/6$)-spectral sparsifier $\bar{H} = (V, \bar{F}, \bar{z})$ of $H$ such that $|\bar{F}| = O\left(\frac{n^2}{\eps^2} \log^3\frac{n}{\eps}\right)$ with high probability, implying the existence of such a sparsifier $\bar{H}$.
    Since $(1 \pm \eps/6)^2 \in [1 - \eps/2, 1 + \eps/2]$, we have $(1-\eps/2)\kappa_{\bar{H}}(X) \le \sfn(X) \le (1+\eps/2)\kappa_{\bar{H}}(X)$ for all $X \subseteq V$.

    We then multiply weights of all hyperarcs in $\bar{H}$ by $\left(1-\varepsilon/2\right)$ to obtain $\tilde{H}=(V,\tilde{F},\tilde{z})$.
    Since $(1 - \eps/2)(1\pm \eps/2)\in [1 - \eps, 1+\eps]$, we have $(1-\eps)\kappa_{\tilde{H}}(X) \le \sfn(X) \le (1+\eps)\kappa_{\tilde{H}}(X)$ for every $X \subseteq V$.
    Also, $0 \le \kappa_{\tilde{H}}(X) = (1 - \eps/2)\kappa_{\bar{H}}(X)\leq (1 - \eps/2)\frac{\sfn(X)}{(1 - \eps/2)}=\sfn(X)\leq 1$ holds for every $X\subseteq V$.

    Recall that {\DHSparsify} performs at most $T = \ceil*{\log_{4/3} \prn*{ \frac{|F|\eps^2}{n^2\log^3(n/\eps)} }} \le 2.5\log_2 n + O(1)$ iterations and that weights of hyperarcs are doubled in each iteration.
    Thus, the weights of $\bar{H}$ are bounded by $\bar{z}_{f} \le 2^T z_f \le 2^{\log_2n^{2.5} + O(1)}n^3 = O(n^{5.5})$ for every $f \in \bar{F}$;
    this also bounds $\ztl_f$ for every $f \in \Ftl$ due to the construction of $\Htl$ from $\bar{H}$.
\end{proof}

We now discuss agnostic learning of set functions on $2^V$.
Let $\Dcal$ be any distribution on $2^V\times [0,1]$ and $\Ccal$ be a class of $[0,1]$-valued functions on $2^V$.
A class $\Ccal$ is said to be \textit{agnostically learnable} if, for any $\eps > 0$ and $\delta \in (0,1)$, there exists an algorithm that, using a sampling oracle from $\Dcal$, returns a hypothesis function $g\colon 2^V \to [0,1]$ such that
    \[
        \E_{(X,y)\sim \Dcal}\left[|g(X)-y|\right] \leq \inf_{g'\in \Ccal}\E_{(X,y)\sim \Dcal}\left[|g'(X)-y|\right] + \eps
    \]
with probability at least $1-\delta$.
The following fact is a well-known consequence of the Chernoff bound and the union bound.
\begin{proposition}[{\citep{soma2019spectral}}]\label[proposition]{proposition:learnable}
    A finite class $\mathcal{C}$ of $[0,1]$-valued functions is agnostically learnable with $O\left(\log (|\mathcal{C}|/\delta)/\eps^2\right)$ samples.
\end{proposition}
By using \cref{lemma:ours-hypernetwork-approx} and \cref{proposition:learnable}, we prove \cref{theorem:submodular} in the same manner as in \citep{soma2019spectral}.
\begin{proof}[Proof of \cref{theorem:submodular}]
    Let $M = O\left(\frac{n^2}{\eps^2} \log^3\frac{n}{\eps}\right)$ and $W = O(n^{5.5})$ denote the maximum hypergraph size and hyperarc weight, respectively, that appear in \cref{lemma:ours-hypernetwork-approx}.
    Let $\Ccal'$ be the class of cut functions of directed hypergraphs $H' = (V, F', z')$ satisfying the following conditions:
    $|F'| \le M$, $z'_f$ is a multiple of $\frac{\eps}{M}$ and $z'_f \le W$ for every $f \in F'$, and $0\leq \kappa_{H'}(X)\leq 1$ for every $X \subseteq V$.
    \cref{lemma:ours-hypernetwork-approx} ensures that for any $g \in \Ccal$ there exists $g' \in \Ccal'$ such that $|g(X) - g'(X)| \le \eps$
    for all $X\subseteq V$.
    Therefore, we only need to show that $\Ccal'$ is agnostically learnable with the desired sample size.
    Since the number of all possible hyperarcs is $4^n-1$ and there are $WM/\eps$ possible weights for each hyperarc, we have
    \begin{align}
        |\mathcal{C'}| \leq O\prn*{\sum_{i=0}^M \binom{4^n}{i}\left(\frac{WM}{\eps}\right)^i} \le O\prn*{M4^{nM} \prn*{\frac{WM}{\eps}}^M},
    \end{align}
    and hence
    \begin{align}
        \log |\mathcal{C'}|= O\left(nM + M \log \left(\frac{WM}{\eps}\right)\right)
        =O\left(\frac{n^3}{\eps^2}\log^3 \frac{n}{\eps}+\frac{n^2}{\eps^2}\log^4 \frac{n}{\eps}\right).
    \end{align}
    Combining this with \cref{proposition:learnable}, we obtain the desired sample complexity bound.
    %\rr{The concrete algorithm for constructing submodular functions follows from that of \citep{soma2019spectral} and \DHSparsify, although it takes an exponential time in $n$ in general.}
\end{proof}

\section{Spectral Sparsification of Undirected Hypergraphs}\label{section:UHS}
To demonstrate the power of the spanner-based sparsification by~Koutis~and~Xu~\citep{koutis2016simple}, we study a natural extension of their algorithm for undirected graphs to undirected hypergraphs.
We show that the resulting algorithm constructs an $\eps$-spectral sparsifier of size $O\prn*{\frac{nr^3}{\eps^2}\log^2 n}$ with high probability and that it also enjoys several useful extensions.

\subsection{Preliminaries}
To begin with, we present an additional background used in this section.

\paragraph*{Notation and Definitions}
We usually denote an ordinary undirected graph and an undirected hypergraph by $G=(V, E, w)$ and $H=(V, F, z)$, respectively, to avoid confusion.
Given an undirected hypergraph $H=(V,F,z)$, a \textit{clique} of $f \in F$ is defined as an edge set $C(f) = \Set*{\{u, v\}}{u, v \in f, u\neq v}$.
(Although a clique conventionally refers to a vertex set, we here regard it as an edge set for convenience.)
For any subset $F'\subseteq F$, we let $C(F') = \bigcup_{f \in F'}C(f)$.

\paragraph*{Spectral Properties of Graphs and Hypergraphs}
We briefly describe the spectral properties of graphs and hypergraphs.
We recommend~\citep{vishnoi2013Lxb,teng2016scalable,spielman2019spectral} for more information on the spectral graph theory and \citep{chan2018spectral,chan2019diffusion} for more details on spectral properties of hypergraphs.

Let $G = (V, E, w)$ be an ordinary graph and $w_e = w(e)$ denote the weight of $e \in E$.
The Laplacian operator $L_G$ of $G$ is written as an $n\times n$ positive semidefinite matrix, called the Laplacian matrix of $G$, such that each $(u, v)$ entry is given by $-w_e$ if $e = \{u, v\} \in E$ and $u \neq v$, $\sum_{e \in E: v \in e} w_e$ if $u = v$, and $0$ otherwise.
The Laplacian matrix $L_G$ can be written as a sum of edge-wise Laplacian matrices $L_e$, i.e., $L_G = \sum_{e \in E} L_e$, where for each $e = \{u, v\} \in E$, the $(u, u)$ and $(v, v)$ entries of $L_e$ are $w_e$, the $(u, v)$ and $(v, u)$ entries are $-w_e$, and the others are zero.
We can also write the quadratic form $x^\top L_G x$ for any $x \in \R^V$ as a sum of edge-wise quadratic forms, i.e.,
$x^\top L_G x = \sum_{e \in E} x^\top L_e x = \sum_{e = \{u, v\}\in E} w_e(x_u - x_v)^2$.

For an undirected hypergraph $H=(V,F,z)$, the Laplacian $L_H\colon \R^V \to \R^V$ is defined as a nonlinear operator that satisfies $x^\top L_H(x) = \sum_{f \in F}z_f\max_{u, v \in f} (x_u - x_v)^2$ for all $x \in \R^V$.
For each $f \in F$, the contribution of $f$ to $x^\top L_H(x)$ is $Q_H^x(f) = z_f\max_{u,v\in f}(x_u-x_v)^2$, which we call the \textit{energy} of $f$.
Note that we have $x^\top L_H(x) = \sum_{f \in F} Q_H^x(f)$.
For any subset $F' \subseteq F$, we let $Q_H^x(F') = \sum_{f \in F'} Q_H^x(f)$, i.e., the sum of energies over $F'$.

An important notion in spectral graph theory is the effective resistance.
Given an ordinary graph $G=(V,E,w)$, the {\em effective resistance} of a pair of vertices $\{u, v\}\ (u,v \in V)$ is given by
\[
    R_G(u,v) = \left(\ones_u - \ones_v\right)^\top L_G^+\left(\ones_u - \ones_v\right),
\]
where $L_G^+$ is the Moore--Penrose pseudo-inverse of $L_G$ and, for each $v \in V$, $\ones_v \in \set*{0,1}^V$ is a vector of all zeros but a single $1$ at the coordinate corresponding to $v$.
The following well-known fact provides another useful characterization of the effective resistance.
\begin{proposition}\label[proposition]{theorem:PRE-EffectiveResistance}
    The effective resistance $R_G(u,v)$ can be defined alternatively as
    \begin{align}\nonumber
        R_G(u,v) = \max_{x\in \R^V} \frac{(x_u-x_v)^2}{x^\top L_G x}.
    \end{align}
\end{proposition}

\subsection{Main Result}
Before proceeding to the main result, we mention the following fact due to \citep{bansal2019new}, which enables us to focus on the case where every hyperedge has a size between $\frac{r}{2}$ and $r$.

\begin{proposition}[\citep{bansal2019new}]\label[proposition]{theorem:UHS-limitr}
Let $H$ be an undirected hypergraph with $n$ vertices and the rank at most $r$, $\eps\in (0,1)$.
For each $i=1,2,\cdots,\lceil\log_2 r\rceil$, let $\tilde{H}_i$ be an  $\eps\sqrt{\frac{2^{i-1}}{r}}$-spectral sparsifier of a sub-hypergraph of $H$ consisting of hyperedges of size $(2^{i-1},\ 2^i]$.
Then, the union of $\tilde{H}_i$ for $i=1,2,\cdots,\lceil\log_2 r\rceil$ is an $\eps$-spectral sparsifier of $H$ with $O\prn*{\frac{nr^3}{\eps^2}\log^2 n}$ hyperedges.
Hence, if there is an algorithm that, given a hypergraph with hyperedges of size $({r}/{2}, r]$, constructs an $\eps$-spectral sparsifier with $O\prn*{\frac{nr^3}{\eps^2}\log^2 n}$ hyperedges, then there is an algorithm that returns an $\eps$-spectral sparsifier with $O\prn*{\frac{nr^3}{\eps^2}\log^2 n}$ hyperedges for any hypergraph of rank at most $r$.
\end{proposition}

Building on this fact, we will prove the following theorem.
\begin{theorem}\label[theorem]{theorem:UHS-main}
    Let $H=(V,F,z)$ be an undirected hypergraph with $|V|=n$ vertices and hyperedges of size between $\frac{r}{2}$ and $r$, and let $\eps \in (0, 1)$.
    Then, {\rm \UHSparsify$(H,\eps)$} given in \cref{alg:UHS-iterative} returns an $\epsilon$-spectral sparsifier $\Htl=(V,\tilde{F},\tilde{z})$ of $H$ satisfying $|\tilde{F}| = O\left(\frac{nr^3}{\eps^2}\log^2n\right)$ with probability at least $1-O\left(\frac{1}{n}\right)$.
\end{theorem}

The size of a sparsifier given in \cref{theorem:UHS-main} matches that of Bansal~et~al.~\citep{bansal2019new} up to a $\log n$ factor, which is the current best upper bound if $r$ is constant.
We present a comparison with existing results in \cref{table:Intro-previous-studies-undirected}.
Furthermore, our algorithm has advantages in the computation complexity, parallel computability, and fault tolerance, as shown in \cref{subsec:UHS_total_complexity,subsubsec:UHS_parallel,subsubsec:UHS_fault_tolerant}, respectively.

\begin{table}[t]\centering
    \caption{Bounds on sparsification of undirected hypergraphs. In the time complexity, additive $\poly(n, 1/\eps)$ terms are omitted.}
    \label{table:Intro-previous-studies-undirected}
    \begin{tabular}{cccc} \toprule
    Method & Cut/Spectral & Bound & Time complexity \\ \midrule
    Newman~and~Rabinovich~\citep{newman2013multiplicative} & Cut & $O(n^2/\eps^2)$ & -{\footnotemark} \\
    Kogan~and~Krauthgamer~\citep{kogan2015sketching} & Cut & $O(n(r+\log n)/\eps^2)$ & $O(mn^2)$  \\
    Chekuri~and~Xu~\citep{chekuri2017note} & Cut & $O(nr(r+\log n)/\eps^2)$ & $\tilde{O}(mr^2)$  \\
    Soma~and~Yoshida~\citep{soma2019spectral}& Spectral & $O(n^3\log n/\eps^2)$ & $O(mr^2)$  \\
    Bansal~et~al.~\citep{bansal2019new} & Spectral & $O(nr^3\log n/\eps^2)$ & $\tilde{O}(mr)${\footnotemark} \\
    Chen~et~al.~\citep{chen2020near} & Cut & $O(n\log n/\eps^2)$ & $\tilde{O}(mn)$  \\
    Kapralov~et~al.~\citep{kapralov2021towards} &  Spectral & $\tilde{O}(nr/\eps^{O(1)})$ & $O(mr^2)$  \\
    Kapralov~et~al.~\citep{kapralov2021spectral} & Spectral & $\tilde{O}(n/\eps^4)$ & $\tilde{O}(mr)$  \\
    Rafiey~and~Yoshida~\citep{rafiey2022sparsification} & Cut & $O(n^2r^2/\eps^2)$ & $O(m2^r)$  \\
    This paper  &  Spectral & $O(nr^3\log^2n/\eps^2)$ & $O(mr)$ \\
    Jambulapati et al.~\citep{Jambulapati2022-qh} and Lee~\citep{Lee2022-rp}  &  Spectral & $O(n\log n \log r/\eps^2)$ &  $\tilde{O}(mr)$ \\
    \bottomrule
    \end{tabular}
  \end{table}
  \addtocounter{footnote}{-1}
  \footnotetext[\thefootnote]{The paper only implicitly shows the existence of cut sparsifiers of size $O(n^2/\eps^2)$.}
  \addtocounter{footnote}{+1}
  \footnotetext[\thefootnote]{This is an improved time complexity bound by Kapralov~et~al.~\citep{kapralov2021spectral}. The original time complexity bound by Bansal~et~al.~\citep{bansal2019new} is $O(mr^2)$.}

Our sparsification algorithm for undirected hypergraphs is almost identical to that for directed hypergraphs given in \cref{section:DHS-AlgorithmDescription}.
The only difference is that, instead of constructing a $\lambda$-coreset (in \cref{line:dhone-findcoreset} in {\DHOnestep}), we construct a natural hypergraph counterpart of the spanner.
Since there is a fast algorithm that constructs a sparse spanner in an ordinary graph, it is natural to try to use it for hypergraphs.
We below define a natural hypergraph counterpart of the spanner and explain how to compute it based on the existing spanner construction algorithm for ordinary graphs.

\subsection{Hyperspanners}\label{subsec:UHS-hyperspanners}
Let $G = (V, E, w)$ be an ordinary graph.
For $k\ge1$, a subgraph $G'$ of $G$ is said to be a $k$-\textit{spanner} of $G$ if for any $\{u, v\} \in E$,
there is a $u$--$v$ path on $G'$ whose total length is at most $k$ times of the length of $\{u, v\}$.
For convenience, we below refer to an edge subset as a spanner.
In the context of spectral sparsification, it is convenient to use $1/w_e$ as the length of $e \in E$ when constructing spanners.
Hence, we say that $S \subseteq E$ is a $k$-spanner of $G$ if for any $e = \{u, v\} \in E$, there is a $u$--$v$ path $P \subseteq S$ such that
\begin{equation}\label{eq:spanner}
    \left(\sum_{e'\in P}\frac{1}{w_{e'}}\right) \le k \cdot \frac{1}{w_e}.
\end{equation}

To define a hypergraph counterpart of the spanner, we first define a \textit{hyperpath} and the distance between its endpoints~\citep{gallo1993directed, gao2015dynamic}; then, we define a \textit{hyperspanner}. %knuth1977ageneralization
\begin{definition}
    Let $H=(V, F, z)$ be an undirected hypergraph and $u, v \in V$ be a pair of vertices.
    We call a set of hyperedges $P = \set*{ f_1, f_2, \dots, f_\ell }$ a $u$--$v$ \textit{hyperpath} if the following conditions hold:
    $u \in f_1$, $v \in f_\ell$, and  $f_i \cap f_{i+1} \ne \emptyset$ for $i=1,\dots,\ell-1$.
    The distance between $u$ and $v$ along $P$ is defined by $\sum_{i=1}^\ell \frac{1}{z_{f_i}}$.
\end{definition}
\begin{definition}
    Let $H=(V, F, z)$ be an undirected hypergraph.
    We say that $S \subseteq F$ is a $k$-\textit{hyperspanner} of $H$ for some $k \ge 1$ if for any $f \in F$ and $\{u, v\} \in C(f)$, there is a $u$--$v$ hyperpath $P \subseteq S$ such that
    \begin{equation}\label{eq:hyper_spanner}
        \prn*{\sum_{f' \in P} \frac{1}{z_{f'}}} \le k \cdot \frac{1}{z_f}.
    \end{equation}
\end{definition}

As we will see shortly, we can easily construct a hyperspanner by looking at the associated graph introduced in~\citep{bansal2019new}.
For a hypergraph $H = (V,F,z)$, the \textit{associated graph} $\GH=(V,E,w)$ of $H$ is a multi-graph obtained from $H$ by replacing each hyperedge $f \in F$ with a clique $C(f)$ with $\binom{|f|}{2}$ edges.
By this definition, each edge $e$ in $\GH$ is associated with a hyperedge $f_e$ in $H$. We define the weight $w_e$ of $e\in E$ by $w_e=z_{f_e}$.

\begin{lemma}\label[lemma]{lem:hyperspanner}
Let  $S_G$ be a $k$-spanner of the associated graph $\GH$ of an undirected hypergraph $H=(V,F,z)$.
Let $S=\Set*{f_e}{e\in S_G}$ be the set of hyperedges associated with the edges in $S_G$.
Then $S$ is a $k$-hyperspanner of $H$.
\end{lemma}
\begin{proof}
    Consider any $f\in F$ and $\{u,v\}\in C(f)$.
    By the definition of associated graphs, $\GH$ has an edge $e$ between $u$ and $v$ that is associated with $f$.
    Since $S_G$ is a $k$-spanner, $S_G$ contains a $u$--$v$ path $P$ in $\GH$ satisfying \cref{eq:spanner}.
    Let $P=\{e_1,\dots, e_{\ell}\}$.
    Although $f_{e_i}=f_{e_j}$ may occur for some $i, j$ with $i\neq j$, the {\em set} $\{f_{e_1}, \dots, f_{e_{\ell}}\}$ contains a $u$--$v$ hyperpath satisfying \cref{eq:hyper_spanner}.  Thus, $S$ is a $k$-hyperspanner.
\end{proof}

Our algorithm uses a bundle of $\lambda$ disjoint ($\log n$)-hyperspanners, which can be obtained by computing ($\log n$)-hyperspanners repeatedly.
Formally, a set $S$ of hyperedges in $H=(V,F,z)$ is called a {\em bundle of $\lambda$ disjoint $k$-hyperspanners} if
$S$ can be written as $S=\bigcup_{j=1}^{\lambda} S_j$,
where $S_j$ is a $k$-hyperspanner of a graph $(V,E\setminus\bigcup_{i=1}^{j-1} S_i,w)$ for each $j$ with $1\leq j\leq \lambda$.
By definition, if $S$ is a bundle of $\lambda$ disjoint $k$-hyperspanners, then $S$ contains $\lambda$ disjoint $u$--$v$ hyperpaths $P_1,\dots, P_{\lambda}$ satisfying \cref{eq:hyper_spanner} for any $f\in F\setminus S$ and $\{u,v\}\in C(f)$.

By \cref{lem:hyperspanner}, we can compute a bundle of $\lambda$ disjoint ($\log n$)-hyperspanners by calling the following algorithm for ordinary graphs repeatedly $\lambda$ times.
\begin{proposition}[\citep{roditty2011dynamic}]\label[proposition]{proposition:graphspanner}
    There is an algorithm that, given a multi-graph with $n$ vertices and $m$ edges,\footnote{Although graphs are assumed to be simple in \citep{roditty2011dynamic}, the result is valid for any multi-graph since only a shortest edge in each parallel class is needed for constructing a spanner. That is, if a graph contains parallel edges, we first choose a shortest edge in each parallel class and then use the algorithm in \citep{roditty2011dynamic}.} computes a ($\log n$)-spanner with $\UCa n$ edges in $\Otl(n^2+m)$ time, where $\UCa$ is an absolute constant.
\end{proposition}

\subsection{Algorithm Description}

\begin{algorithm}[htb]
    \caption{\UHOnestep$(H, \lambda)$: sampling algorithm called in each iteration in \cref{alg:UHS-iterative}.}\label{alg:UHS-onestep}
    \begin{algorithmic}[1]
        \Require $H = (V, F ,z)$ and $\lambda > 0$
        \Ensure $\Htl = (V, \Ftl, \ztl)$
            \State Compute a bundle $S$ of $\lambda$ disjoint ($\log n$)-hyperspanners of $H$
            %\State $S \gets \bigcup_{i=1}^{\lambda} S_i$
            \State $\Ftl \gets S$ and $\ztl_f \gets z_f$ for $f \in S$
            \For{each $f \in F\setminus S$}
            \State{With probability $\frac12$, add $f$ to $\tilde{F}$ and set $\tilde{z}_f \gets 2z_f$}
            \EndFor
            \State \Return $\Htl = (V,\tilde{F}, \tilde{z})$
    \end{algorithmic}
\end{algorithm}
\begin{algorithm}[htb]
    \caption{\UHSparsify$(H, \eps)$: iterative algorithm that computes an $\eps$-spectral sparsifier.}\label{alg:UHS-iterative}
    \begin{algorithmic}[1]
        \Require $H = (V,F,z)$ with $|V|=n$ and $|F|=m$ and $\eps > 0$
        \Ensure $\Htl = (V,\Ftl, \ztl)$
        \State $m^* \gets \frac{nr^3}{\eps^2}\log^2n$ \Comment{This is the (asymptotic) target size of the resulting sparsifier.  }
        \State $T \gets \ceil*{\log_{4/3} \left(\frac{m}{m^*}\right)}$
        \State $i \gets 0$, $\Htl_0 = (V, \Ftl_0, \ztl_0) \gets H$, and $m_0 \gets |\Ftl_0|$
        \While{$i < T$ and $m_i \geq \UCc m^*$} \Comment{$\UCc$ is a constant explained in \cref{subsection:UHS_proof_iterative}.}
        \State $\eps_i \gets \frac{\eps}{4 \log_{4/3}^2\left(\frac{m_i}{m^*}\right)}$ and $\lambda_i \gets \ceil*{ \frac{8\UCb r^3\log^2 m_i}{\eps_i^2} }$ {\label[line]{line:UHSparsify-eps-lambda}}
        \Comment{$\eps_i$ is used in the analysis.}
        \State $\Htl_{i+1} = (V,\tilde{F}_{i+1},\tilde{z}_{i+1}) \gets \text{\UHOnestep}(\Htl_{i}, \lambda_i)$
        \State $m_{i+1} \gets |\tilde{F}_{i+1}|$
        \State $i\gets i+1$
        \EndWhile
        \State $\Tend\gets i$ and $\Htl\gets \Htl_{\Tend}$
        \State \Return $\Htl = (V,\tilde{F}, \tilde{z})$
    \end{algorithmic}
\end{algorithm}

We give a description of our algorithm for undirected hypergraphs in \cref{alg:UHS-onestep,alg:UHS-iterative}.
{\UHSparsify} (\cref{alg:UHS-iterative}) is identical to {\DHSparsify}  (\cref{alg:DHS-iterative}) for  directed hypergraphs except for the choice of $m^*$ and $\lambda^*$,
and {\UHOnestep} (\cref{alg:UHS-onestep}) is identical to  {\DHOnestep} (\cref{alg:DHS-onestep})
except for the first line, where {\UHOnestep} constructs a bundle of $\lambda$ hyperspanners instead of a $\lambda$-coreset.
Our goal is to prove that {\UHSparsify} returns an $\varepsilon$-sparsifier with $O\prn*{\frac{nr^3}{\varepsilon^2}\log^2 n}$ hyperedges.

{\UHSparsify} iteratively calls {\UHOnestep}, which first computes a bundle of $\lambda$ disjoint ($\log n$)-hyperspanners and then samples remaining hyperedges with probability ${1}/{2}$.
Since we can use the same argument as that in \cref{section:DHS-NearlyTightSparsification} to analyze {\UHSparsify}, the only nontrivial part is to show the correctness of {\UHOnestep} called in each iteration.
That is, our main goal is to prove the following lemma, for which we will use the fact that the output of {\UHOnestep} always contains a bundle of disjoint hyperspanners.
\begin{lemma}\label[lemma]{lemma:UHS-onestep}
    Let $H=(V,F,z)$ be an undirected hypergraph with $n = |V|$ vertices and $m = |F|$ hyperedges, where every hyperedge $f \in F$ satisfies $|f| \in (r/2, r]$.
    For any $\eps \in (0, 1)$ and $\lambda\geq \frac{8\UCb{r}^3 \log^2 n}{\eps^2}$, where $\UCb$ is a sufficiently large constant, {\UHOnestep}$(H, \lambda)$ returns an $\epsilon$-spectral sparsifier $\Htl=(V,\Ftl,\ztl)$ of $H$ satisfying $|\Ftl|\leq \frac{m}{2} + \prn*{3m\log n}^{\frac12} + \lambda \UCa n$ with probability at least $1-O\left(\frac{1}{n^2}\right)$.
\end{lemma}
The constant $\UCa$ is the one used as a factor of the spanner size in \cref{proposition:graphspanner}, and $\UCb$ is a constant coming from \cref{prop:bansal} (explained below).

We first prove \cref{lemma:UHS-onestep} in \cref{subsec:UHS_proof_onestep} and then \cref{theorem:UHS-main} in \cref{subsection:UHS_proof_iterative}.

\subsection{Proof of \texorpdfstring{\cref{lemma:UHS-onestep}}{Lemma~\ref{lemma:UHS-onestep}}}\label{subsec:UHS_proof_onestep}
In this section, let $H=(V,F,z)$, $\lambda$, and $\eps$ be as given in the statement of  \cref{lemma:UHS-onestep}, and let $\Htl=(V,\Ftl,\ztl)$ be the output of \UHOnestep$(H,\lambda)$.
We also use $G = (V,E,w)$ to denote the associated graph of $H$.

First, analogous to \cref{lemma:DHS-size}, we can bound the size of $\Htl$ with high probability.

\begin{lemma}\label[lemma]{lemma:UHS-boundedgesize}
    Let $H=(V,F,z)$ be an undirected hypergraph with $|V| = n$ and $|F| = m$ and $\lambda$ be a positive integer.
    If $m > 12\log n$,\footnote{Otherwise, a given graph is already sparsified and we do not run {\UHOnestep}.}
    {\UHOnestep}$(H,\lambda)$ outputs a sub-hypergraph $\Htl=(V,\Ftl,\ztl)$ of $H$ satisfying
    $|\Ftl|\leq \frac{m}{2} + \prn*{3m\log n}^{\frac12} + \lambda \UCa n$ with probability at least $1-\frac{2}{n^2}$.
\end{lemma}
\begin{proof}
    Replacing $\lambda n^2$ in the proof of \cref{lemma:DHS-size} with $\lambda \UCa n$ yields the desired bound.
\end{proof}

We below bound the sparsification error of {\UHOnestep}.
We build on the proof strategy of Bansal~et~al.~\citep{bansal2019new} and show that the sampling probability in {\UHOnestep} satisfies a sufficient condition for a sampling algorithm to produce an $\eps$-spectral sparsifier (see \cref{prop:bansal}).

The following lemma ensures that once $\lambda$ hyperspanners are selected, the remaining hyperedges have small effective resistances in the associated graph.
\begin{lemma}\label[lemma]{lemma:UHS-onestep-EF}
    Let $H=(V,F,z)$ be an undirected hypergraph such that $|V| = n$ and $|f| \in (r/2, r]$ for every $f \in F$, $G = (V,E,w)$ be the associated graph of $H$, $\lambda$ be a positive integer, and $S$ be a bundle of $\lambda$ disjoint  ($\log n$)-hyperspanners.
    %$S$ be the output of {\HyperSpannerFinder}$(H,\lambda)$.
    For any hyperedge $f\in F\setminus S$ and $\{u,v\}\in C(f)$, we have $z_f R_G(u,v) \leq \frac{4\log n}{r\lambda}$.
%    \begin{align}
%        \label{eq:UHS-onestep-EFbound}
%        z_f R_G(u,v) \leq \frac{4\log n}{r\lambda}.
%    \end{align}
\end{lemma}
\begin{proof}
Fix any $f\in F\setminus S$ and $\{u,v\}\in C(f)$.
From \cref{theorem:PRE-EffectiveResistance},
%for any $\{u', v'\} \in C(F)$,
we have $R_G(u,v) = \max_{x \in \mathbb{S}_G} (x_{u}-x_{v})^2$,
where $\mathbb{S}_G = \Set*{x \in \R^V}{x^\top L_G x = 1}$.
Hence, it suffices to prove that $z_f (x_u-x_v)^2\leq \frac{4\log n}{r\lambda}$ holds for any $x \in \mathbb{S}_G$.

We first show that any $f' \in F$ and $a, b \in C(f')$ satisfy
\begin{align} \label{eq:UHS-onestep-oneclique}
    \sum_{\{a', b'\}\in C(f')} z_{f'} (x_{a'}-x_{b'})^2\geq  \frac{r}{4} z_{f'} (x_{a}- x_{b})^2.
\end{align}
This is trivial if $|f'| = 2$ since $C(f') = \set{a, b}$ and $r<4$ due to $|f'| \in (r/2, r]$.
Otherwise, $C(f')$ has $|f'|-2$ vertices other than $a$ and $b$.
Let $c$ be one such vertex, and consider the two edges $\{a,c\}, \{b,c\}\in C(f')$ in the associated graph.
Since $(x_a-x_c)^2 + (x_b-x_c)^2 \geq  \frac{1}{2} (x_{a}- x_b)^2$, summing both sides over all $c\in f'\setminus \{a,b\}$ yields
\begin{align}
    \sum_{c\in f'\setminus \{a,b\}} \prn*{(x_a-x_c)^2 + (x_b-x_c)^2} \geq  \frac{|f'|-2}{2}  (x_{a}- x_b)^2.
\end{align}
Since the size of each hyperedge is assumed to be in $(r/2, r]$, we have $|f'| \ge r/2$, hence
\begin{align}
    %\therefore
    \sum_{\{a' ,b'\}\in C(f')} (x_{a'} - x_{b'})^2
    & \geq
    (x_{a}- x_b)^2 + \sum_{c\in f'\setminus \{a,b\}} \prn*{(x_a-x_c)^2 + (x_b-x_c)^2}
    \\
    & \geq
    \frac{|f'|}{2}  (x_{a}- x_b)^2
    \\
    & \geq
    \frac{r}{4}  (x_{a}- x_b)^2.
\end{align}
Multiplying $z_{f'}$ to both sides yields \cref{eq:UHS-onestep-oneclique}.

Since $S$ is a bundle of $\lambda$ disjoint  ($\log n$)-hyperspanners,
$S$ contains $\lambda$ disjoint $u$--$v$ hyperpaths $P_i = \{f^i_{j}\}_{j=1}^{\ell_i}$ $(i=1,2,\cdots,\lambda)$ such that
\begin{align}\label{eq:UHS-onestep-pathcondition}
    z_f \leq \log n\left(\sum_{j=1}^{\ell_i} z_{f^i_{j}}^{-1}\right)^{-1}.
\end{align}
Take one such hyperpath $P_i$.
Let $v^i_0 = u$, $v^i_{\ell_i} = v$, and $v^i_j$ be a vertex contained in $f^i_{j} \cap f^i_{j+1}$ ($j=1,2,\cdots,\ell_i-1$).
The Cauchy--Schwarz inequality of two vectors,
\begin{align}
    &
    \left(
    z_{f^{i}_{1}}^{1/2} \left( x_{v_{0}^{i}} - x_{v_1^i} \right),
    z_{f^{i}_{2}}^{1/2} \left( x_{v_{1}^{i}} - x_{v_2^i} \right),
    \dots,
    z_{f^{i}_{\ell_i}}^{1/2} \left( x_{v^{i}_{\ell_i-1}} - x_{v^i_{\ell_i}} \right)
    \right)
    \quad
    \text{and}
    \\
    &
    \left(
    z_{f^i_1}^{-1/2},
    z_{f^i_2}^{-1/2},\cdots,
    z_{f^i_{\ell_i}}^{-1/2}
    \right),
\end{align}
implies
\begin{align}
    \left(
    \sum_{j=1}^{\ell_i}
        z_{f^i_j} \left(x_{v^{i}_{j-1}}-x_{v^i_j}\right)^2
    \right)
    \left(
    \sum_{j=1}^{\ell_i}
        z_{f^i_j}^{-1}
    \right)
    &\geq
    \left(
    \sum_{j=1}^{\ell_i}
        z_{f^i_j}^{1/2} \left(x_{v^{i}_{j-1}}-x_{v^i_j}\right) \cdot z_{f^i_j}^{-1/2}
    \right)^2
    \\
    &=
    \left(
    \sum_{j=1}^{\ell_i}
        \left(x_{v^{i}_{j-1}}-x_{v^i_j}\right)
    \right)^2
    \\
    &=
    \left(
        x_{v^{i}_{0}}-x_{v^i_{\ell_i}}
    \right)^2
    \\
    &=
    \left(
        x_u - x_v
    \right)^2.
\end{align}
By using this and \cref{eq:UHS-onestep-pathcondition}, we obtain
\begin{align}\label{eq:UHS-onestep-onespanner}
    \sum_{j=1}^{\ell_i}
        z_{f^i_j} \left(x_{v^{i}_{j-1}}-x_{v^i_j}\right)^2
    \geq
    \left(
    \sum_{j=1}^{\ell_i}
        z_{f^i_j}^{-1}
    \right)^{-1}
    \left(
        x_u - x_v
    \right)^2
    \geq
    \frac{z_f}{\log n}
    \cdot
    \left(
        x_u - x_v
    \right)^2
    .
\end{align}
Combining \cref{eq:UHS-onestep-oneclique,eq:UHS-onestep-onespanner}, for any $x \in \mathbb{S}_G$, we have
\begin{align}
    1 = x^\top L_G x =
    \sum_{e=\{a,b\}\in E} w_e (x_a-x_b)^2
    &
    \geq
    \sum_{i=1}^\lambda \sum_{j=1}^{\ell_i} \sum_{\{a,b\}\in C(f^i_j)} z_{f^i_j} (x_a-x_b)^2
    \\
    & \geq
    \sum_{i=1}^\lambda \sum_{j=1}^{\ell_i} \frac{rz_{f^i_j} \left(x_{v^i_{j-1}}-x_{v^i_j}\right)^2 }{4}
    \quad (\text{by \cref{eq:UHS-onestep-oneclique}})
    \\
    & \geq
    \sum_{i=1}^\lambda \frac{rz_f \left(x_{u}-x_{v}\right)^2   }{4\log n}
    \quad (\text{by \cref{eq:UHS-onestep-onespanner}})
    \\
    & =
    \frac{\lambda r }{4\log n}  z_f \left(x_{u}-x_{v}\right)^2.
\end{align}
From this inequality and \cref{theorem:PRE-EffectiveResistance}, we obtain $z_f R_G(u,v) \leq \frac{4\log n}{\lambda r}$ as desired.
%    \begin{align}
%        z_f R_G(u,v)
%        %=
%        %z_f \max_{x \in \mathbb{S}_G} (x_u-x_v)^2
%        \leq
%        \frac{4\log n}{\lambda r},
%    \end{align}
%    thus completing the proof.
\end{proof}

The following claim due to Bansal~et~al.~\citep{bansal2019new} provides a sufficient condition for a sampling probability to yield an $\eps$-spectral sparsifier with high probability.

\begin{proposition}[{\citep{bansal2019new}}]\label[proposition]{prop:bansal}
Let $H=(V,F,z)$ be an undirected hypergraph with $|V| = n$ and $|f| \in (r/2, r]$ for every $f \in F$ and $\GH$ be the associated graph of $H$.
For each $f \in F$, let $p_f$ be a number satisfying $p_f \in [0,1]$ and
\begin{align}
    \label{eq:BansalSamplingProbability}
    p_f \geq \min\left\{
        1, \max_{\{u,v\}\in C(f)} \frac{\UCb r^4 \log n}{\eps^2} \cdot z_f R_G(u,v)
    \right\},
\end{align}
where $\UCb$ is an absolute constant.
Then, by sampling each $f \in F$ independently with probability $p_f$ and setting its weight to $z_f / p_f$ if sampled, we can obtain an $\eps$-spectral sparsifier of $H$ with probability at least $1-O(1/n^2)$.
\end{proposition}
%\begin{remark}
It should be noted that the original version in \citep{bansal2019new} deals with the case when
the equality holds in \cref{eq:BansalSamplingProbability}.
The above extended version is given in \citep{kapralov2021spectral}.

We below show that, if $\lambda$ is sufficiently large as assumed in the statement of \cref{lemma:UHS-onestep},  the sampling probability in
{\UHOnestep}$(H,\lambda)$ satisfies \cref{eq:BansalSamplingProbability}.
This fact, together with \cref{prop:bansal}, completes the analysis of the sparsification error of {\UHOnestep}.
\begin{proof}[Proof of \cref{lemma:UHS-onestep}]
Let $S$ be the bundle of the hyperspanners constructed in the first line in {\UHOnestep}.
Note that we can regard {\UHOnestep} as an algorithm that first assigns a sampling probability of $p_f = 1$ (resp.\ $1/2$) to each $f \in S$ (resp.\ $f \in F \setminus S$) and then samples each hyperedge independently with probability $p_f$ (and multiply $z_f$ by $1/p_f$).
Since $p_f = 1$ trivially satisfies \cref{eq:BansalSamplingProbability}, we below discuss the probability of $p_f = 1/2$ for $f \in F \setminus S$.

Let $f \in F\setminus S$ and $\{u,v\}\in C(f)$.
Since we have $\lambda \geq \frac{8 \UCb r^3 \log^2 n}{\eps^2}$, \cref{lemma:UHS-onestep-EF} implies $z_f R_G(u,v) \leq \frac{4\log n}{r\lambda} \leq \frac{\eps^2}{2 \UCb r^4 \log n}$.
%    \begin{align}
%        z_f R_G(u,v) \leq \frac{4\log n}{r\lambda} \leq \frac{\eps^2}{2 \UCb r^4 \log n}
%    .
%    \end{align}
Hence, for any $\{u,v\}\in C(f)$, it holds that
\begin{align}
\frac12
=
\frac{\UCb r^4 \log n}{\eps^2} \cdot \frac{\eps^2}{2 \UCb r^4 \log n}
\ge
\frac{\UCb r^4 \log n}{\eps^2} \cdot z_f R_G(u,v).
\end{align}

Therefore, for every hyperedge $f \in F$, the sampling probability $p_f$ in \UHOnestep$(H, \lambda)$ satisfies \cref{eq:BansalSamplingProbability}.
Consequently, by \cref{prop:bansal}, \UHOnestep$(H, \lambda)$ with $\lambda \geq \frac{8 \UCb r^3 \log^2 n}{\eps^2}$ returns an $\eps$-spectral sparsifier of $H$ with probability at least $1-O({1}/{n^2})$.
Combining this with the size bound (\cref{lemma:UHS-boundedgesize}) completes the proof of \cref{lemma:UHS-onestep}.
\end{proof}

\subsection{Proof of \texorpdfstring{\cref{theorem:UHS-main}}{Theorem~\ref{theorem:UHS-main}}}\label{subsection:UHS_proof_iterative}
In this section, let $H=(V,F,z)$ be an input hypergraph with $|V|=n$ and $|F|=m$, $\eps\in (0,1)$, and $\Htl=(V,\Ftl,\ztl)$ be the output of {\UHSparsify}$(H,\eps)$.
We show that $\Htl$ is an $\eps$-spectral sparsifier of $H$ with $|\Ftl|=O\prn*{\frac{nr^3}{\eps^2}\log^2n}$.
We define $m^*, T, i_{\rm end}, (\Htl_i = (V,\Ftl_i, \tilde{z_i}), \lambda_i),m_i$, and $\eps_i$ as given in the algorithm description of {\UHSparsify}$(H,\eps)$ (\cref{alg:UHS-iterative}).
That is,
$m^*=\frac{nr^3}{\eps^2}\log^2n$ is the target size,
$T = \ceil*{\log_{4/3} \left(\frac{m}{m^*}\right)}$ is the maximum number of iterations,
$\Tend$ is the number of iterations performed,
($\Htl_i = (V,\Ftl_i, \tilde{z_i})$, $\lambda_i$) is the input of {\UHOnestep} at the $i$th iteration, $m_i=|\tilde{F}_i|$,
and $\eps_i = \frac{\eps}{4 \log_{4/3}^2\left(\frac{m_i}{m^*}\right)}$ as in \cref{line:UHSparsify-eps-lambda} of \cref{alg:UHS-iterative}.

We first show that the number of hyperedges decreases geometrically over the iterations.
\begin{lemma}\label[lemma]{lemma:UH-threequarter}
    Let $m_i$ be the number of hyperedges in $\Htl_i$, and
    suppose that $m_i \geq \UCc m^*$ for a sufficiently large constant $\UCc$. Then, we have $\prn*{3m_i\log n}^{\frac12} + \lambda_i \UCa n \leq \frac14 m_i$.
    %\begin{align}
    %    \prn*{3m_i\log n}^{\frac12} + \lambda_i \UCa n \leq \frac14 m_i.
    %\end{align}
\end{lemma}
\begin{proof}
    The proof is almost identical to that of \cref{lemma:DH-threequarter};
    the only difference is that we set $\lambda_i = \ceil*{\frac{8\UCb r^3\log^2n}{\eps_i^2}}$, where $\eps_i = \frac{\eps}{4\log_{4/3}^2\frac{m_i}{m^*}}$.
    We can prove $\prn*{3m_i\log n}^{\frac12} \le \frac{m_i}{8}$ and $\lambda_i \UCa n \le \frac{m_i}{8}$ if $m_i \geq \UCc m^*$ holds for a sufficiently large constant $\UCc$, as in the proof of \cref{lemma:DH-threequarter}.
    Thus we obtain the claim.
\end{proof}

\begin{proof}[Proof of \cref{theorem:UHS-main}]
    Again, the proof is almost the same as that of \cref{theorem:DHS-main}.
    From \cref{lemma:UHS-onestep,lemma:UH-threequarter}, with probability at least $1 - O(T/n^2) \gtrsim 1 - O(1/n)$, we can suppose that every iteration of {\UHSparsify}$(H, \eps)$ is successful, i.e., $\Htl_{i+1}$ is an $\eps_{i}$-sparsifier of $\Htl_i$ and $m_{i+1} \leq \frac34 m_i$ holds.
    If $m_i \leq \UCc m^* = \UCc \frac{nr^3}{\eps^2}\log^2 n$ for some $i\leq T-1$, the size of $\Htl = \Htl_{\Tend}$ is already small as desired.
    If not, since $m_{i+1} \le \frac34 m_i$,
    we have $m_T \le m\cdot (3/4)^T \le m \cdot (3/4)^{\log_{4/3}\frac{m}{m^*}} = m^*$.
    Thus, the size of $\Htl = \Htl_{\Tend}$ is $O\prn*{ \frac{nr^3}{\eps^2} \log^2 n}$ in any case.
    Also, by the same discussion as the proof of \cref{theorem:DHS-main}, $\Htl = \Htl_{\Tend}$ is an $\eps$-spectral sparsifier of $H$ since $\Htl_{i+1}$ is an $\eps_i$-spectral sparsifier of $\Htl_i$ for $i = 0,1,\dots,\Tend - 1$.
\end{proof}

\subsection{Total Time Complexity}\label{subsec:UHS_total_complexity}
We bound the time complexity of {\UHSparsify}$(H, \eps)$.
\begin{theorem}\label[theorem]{theorem:UHS-Complexity}
    For any undirected hypergraph $H=(V,F,z)$ with $n$ vertices and $m$ hyperedges, where every hyperedge $f \in F$ satisfies $|f| \in (r/2, r]$, and  $\eps \in (0,1)$, {\UHSparsify}$(H, \eps)$ runs in $O(rm+ \poly(n,\eps^{-1}))$ time with probability at least $1 - O(1/n)$.
\end{theorem}

Recall that the trick by Bansal~et~al.~\citep{bansal2019new} (see \cref{theorem:UHS-limitr}) enables us to restrict input hypergraphs to those consisting of hyperedges of size between $r/2$ and $r$.
With this trick, we can extend \cref{theorem:UHS-Complexity} to any undirected hypergraph $H$ of rank $r$.
More precisely, for each $i = 1,\dots,\ceil*{\log_2 r}$, we apply {\UHSparsify} to the sub-hypergraph of $H$ consisting of hyperedges of size $\left(2^{i-1}, 2^i\right]$, whose hyperedge set is denoted by $F_i$,
and return the union of those outputs.
Each {\UHSparsify} computes a desired sparsifier in $O\prn*{|F_i|2^i+ \poly(n,\eps^{-1})}$ time with probability at least $1 - O(1/n)$ by \cref{theorem:UHS-Complexity}.
Hence, the total time complexity is $\sum_{i=1}^{\ceil*{\log_2 r}} O\prn*{|F_i|2^i+ \poly(n,\eps^{-1})} \lesssim O(rm+ \poly(n,\eps^{-1}))$.
By \cref{theorem:UHS-limitr}, the algorithm outputs an $\eps$-spectral sparsifier with probability at least $1 - O(\ceil*{\log_2 r}/n) \gtrsim 1 - \Otl(1/n)$.
%\end{remark}

\begin{proof}[Proof of \cref{theorem:UHS-Complexity}]
We show how to implement {\UHSparsify} so that it runs in  $O(rm + \poly(n,\eps^{-1}))$ time.
To this end, we focus on the computation of a bundle of $\lambda$ disjoint ($\log n$)-hyperspanners since how to implement the other parts in {\UHOnestep} and {\UHSparsify} is trivial.
A native implementation discussed in \cref{subsec:UHS-hyperspanners} is that we construct the associated graph $\GH$ of $H$ and call the spanner-construction algorithm in \cref{proposition:graphspanner}.
This requires $O(\lambda r^2 m)$ time.
We below improve this implementation to compute a bundle of $\lambda$ disjoint hyperspanners in $O(rm+\poly(n,\eps^{-1}))$ time.
We need two ideas to achieve this goal.

\medskip

\noindent \textbf{(1) Replacing cliques with stars.}
We borrow the following idea from \citep{kapralov2021spectral}.
When constructing the associated graph, instead of the clique $C(f)$, we use the edge set $C^*(f)$ of a star graph supported on $f$, choosing one vertex as its center.
Note that $C^*(f)$ is a $2$-spanner of $C(f)$, and thus there is no asymptotic loss in the stretch factor in the resulting spanner.
For any subset $F'\subseteq F$, we let $C^*(F')=\bigcup_{f\in F'}C^*(f)$ and $C^*_+(F') = \biguplus_{f \in F'}C(f)$, where the latter denotes the multiset of edges, for later use.

Let $G^*(H) = (V, E, w)$ be the multi-graph obtained from $H$ by replacing each hyperedge $f$ with $C^*(f)$ and setting $w_e = z_f$ for $e \in C^*(f)$.
Then, a ($\log n$)-spanner in $G^*(H)$ gives a ($2\log n$)-hyperspanner in $H$.
This only doubles the bound on the size of the resulting sparsifier since we just need to increase $\lambda$ to $2\lambda$ to keep the sparsification error within $\eps$.

\medskip

\noindent \textbf{(2)  Eliminating unnecessary hyperedges before constructing hyperspanners.}
Based on (1), our strategy goes as follows.
We first construct $G^*(H)$
and compute a  spanner $T_i$ of size $C_3 n$ in $G^*(H)$ by calling the algorithm in \cref{proposition:graphspanner},
where $C_3$ is the constant given in the statement of \cref{proposition:graphspanner}.
Next, we compute a hyperspanner $S_i$ consisting of hyperedges $f_e$ associated with edges $e \in T_i$ as shown in \cref{lem:hyperspanner}.
After computing $S_i$, we remove edges in $C^*_+(S_i)$ from $G^*(H)$, and repeat this process $\lambda$ times.
This still has room for improvement because $G^*(H)$ contains lots of unnecessary edges.
Specifically, $G^*(H)$ may contain parallel edges between two vertices $u$ and $v$, and we may always suppose that heavier edges are selected in the spanner $T_i$ in each parallel class since a heavier weight implies a shorter distance.
Hence, removing all parallel edges except for the $\lambda C_3 n$ heaviest ones in each parallel class does not change the resulting $\lambda$ spanners $T_1,\dots, T_{\lambda}$.
(Note that, for each $S_i$, $C^*_+(S_i)$ may contain up to $C_3 n$ parallel edges in a parallel class.
Those edges are removed and unavailable when constructing the next spanner $T_{i+1}$.)
For each $\{u,v\}\in C^*(F)$, we can compute $\lambda C_3 n$ heaviest parallel edges in the parallel class between $u$ and $v$ in $O(|F^{uv}|)$ time by using a linear-time selection algorithm,
where $F^{uv}$ denotes the set of hyperedges $f$ with $\{u,v\}\in C^*(f)$.
Since $\sum_{\{u,v\}\in C^*(F)} |F^{uv}|\leq rm$, the total time for this process is $O(rm)$.

The resulting multi-graph has $O(\lambda  n^3)$ edges, and thus a ($\log n$)-spanner in $G^*(H)$ can be computed in
$\tilde{O}(\lambda n^3)$ time by \cref{proposition:graphspanner}.
In total, we can compute a bundle of $\lambda$ disjoint ($2\log n$)-hyperspanners in $H$ in $O(rm+\poly(n,\eps^{-1}))$ time as required.

\vspace{\baselineskip}

With the above implementation, {\UHOnestep}$(H,\lambda)$ runs in $O(rm+\poly(n,\eps^{-1}))$ time.
We now turn to the analysis of the total time complexity.
Recall that {\UHSparsify}$(H,\eps)$ calls {\UHOnestep}$(\tilde{H}_i,\lambda_i)$ for $i$ with $1\leq i\leq i_{\rm end}-1\leq T=\ceil*{ \log_{4/3}\frac{m}{m^*} }$.
Let $m_i$ be the size of $\tilde{H}_i$.
Then, the total time complexity is
$\sum_{i=1}^{T} O(rm_i+\poly(n,\eps^{-1}))$.
We have shown in the proof of \cref{theorem:UHS-main} that
$m_{i+1}\leq \frac{3}{4} m_i$ holds for all $i$ with probability at least $1 - O(T/n^2)$.
Hence, by $m_0=m$, we have $\sum_{i=1}^{T} m_i=O(m)$.
Also, $T=O(n)$ holds.
To conclude, the total time complexity is $O(rm +\poly(n,\eps^{-1}))$ with probability at least $1 - O(1/n)$.
\end{proof}

\subsection{Advantages Inherited from Spanner-based Sparsification}\label{subsec:UHS_advantages_of_spanner}
Spanner-based sparsification algorithms for ordinary graphs are known to enjoy several advantages, such as parallel computability and fault tolerance \citep{koutis2016simple,zhu2019improved}.
We demonstrate that our extension to hypergraphs inherits those advantages.

\subsubsection{Parallel Computability}\label{subsubsec:UHS_parallel}
We discuss the implementation of our algorithm for undirected hypergraphs in the PRAM (parallel random-access machine) CRCW (concurrent read concurrent write) model, one of the basic models for parallel computation.
In this setting, multiple processors can read and write in the same place of shared memory at the same time \citep{jeje1992introduction,reif1993synthesis}.

We below use the fact that we can quickly compute spanners in parallel.
\begin{proposition}[\citep{baswana2007simple}]\label[proposition]{proposition:baswana-spanner}
    Given a multi-graph with $n$ vertices and $m$ edges, in the PRAM CRCW model with $m$ processors, we can compute a ($\log n$)-spanner with at most $\UCd n\log n$ edges in $O(\log^2 n + \log m)$ time with $\tilde{O}(m)$ operations.\footnote{Again, although the paper \citep{baswana2007simple} focuses on the case where graphs are simple, the result is valid for any multi-graph since we can convert it to a simple graph by choosing a shortest edge in each parallel class in $O(\log m)$ time with $m$ processors.}
\end{proposition}

We will also use the fact that we can quickly find the largest $k$ elements from a  set of size $m$ by using $m$ processors.
\begin{proposition}[\citep{dietz1994very}]\label[proposition]{proposition:APC-kth-informal}
    In the PRAM CRCW model with $m$ processors, we can compute the $k$ largest elements in a totally ordered set of $m$ items in $O(\log \log m + \log k)$ time and $O(m)$ operations.\footnote{The original time complexity is $O(\log \log m + \log k/\log \log m)$, but we here omit $1/\log \log m$ for simplicity.}
\end{proposition}

Our goal in this section is to prove the following theorem.
\begin{theorem}\label[theorem]{theorem:UHS-CRCW}
    Let $H=(V,F,z)$ be an undirected hypergraph with $|V| = n$, $|F| = m$, and the rank $r$, and let $\eps\in(0,1)$.
    In the PRAM CREW model with $n^2m$ processors, we can compute an $\eps$-spectral sparsifier of $H$ with $O(r^3n \log^3 n/\eps^2)$ hyperedges with probability at least $1 - \Otl(1/n)$ in $O\prn*{\frac{r^3}{\eps^2}\log^7 m \cdot \prn*{\log^2 n + \log\frac{1}{\eps}}}$ time.
\end{theorem}

This is an NC algorithm if we regard $r$ and $\eps$ as constants.
Developing an NC algorithm for general cases remains open.

We also mention that the sparsifier size in \cref{theorem:UHS-CRCW} has an extra $\log n$ factor compared with the bound in \cref{theorem:UHS-main}.
This extra factor comes from the $\log n$ increase in the size of a ($\log n$)-spanner in the existing parallel algorithm (cf. \cref{proposition:graphspanner,proposition:baswana-spanner}).

\begin{proof}[Proof of \cref{theorem:UHS-CRCW}]
We first discuss the time complexity, assuming that an input hypergraph consists of hyperedges of size in $(r/2, r]$.
At the end of the proof, we show that we can use the trick in \cref{theorem:UHS-limitr} without increasing the asymptotic time complexity.
Since {\UHSparsify}$(H,\eps)$ sequentially calls {\UHOnestep}, we below focus on how to implement {\UHOnestep} in parallel.

Let us consider {\UHOnestep}$(H,\lambda)$.
Since it is easy to sample each hyperedge in parallel, the only nontrivial part is the construction of a bundle of disjoint hyperspanners in the first line in {\UHOnestep}$(H,\lambda)$.
We show that this can be done in $O\prn*{\log m+\lambda \prn*{\log^2 n + \log\frac{1}{\eps}}}$ time using $O(n^2 m)$ processors.

Recall that our original sequential algorithm in \cref{subsec:UHS_total_complexity} consists of the following three steps:
(1) construct an undirected multi-graph $G^*(H)$ by replacing each hyperedge $f$ with a star $C^*(f)$,
(2) remove edges in $G^*(H)$ except for the $\lambda C_3 n$ heaviest ones in each parallel class, and
(3) greedily compute $\lambda$ ($\log n$)-spanners in the resulting multi-graph.

Regarding step (1), we first compute $F^{uv}=\Set*{f}{\{u,v\}\in C^*(f)} \subseteq F$ for each $\{u,v\}$.
This can be done with $n^2m$ processors in $O(\log m)$ time by using the parallel prefix sum computation~\citep{ladner1980parallel,cole1989faster}.
(Specifically, for each $\{u,v\}$, the algorithm first prepares an array of size $m$ that represents whether $u$ and $v$ are contained in $C^*(f)$ for each $f\in F$,
next computes the order of the hyperedges containing $\{u,v\}$ by the prefix sum computation in the array,
and finally computes an array of size $|F^{uv}|$ that stores only the elements in $F^{uv}$ based on the order.)
Each element in $F^{uv}$ can be identified with an edge $e \in G^*(H)$ associated with a hyperedge $f_e$,
and hence the collection of $F^{uv}$ for $\{u,v\}\in C^*(F)$ represents $G^*(H)$.

Step (2) can be done by applying \cref{proposition:APC-kth-informal} to each $\{u,v\}\in C^*(F)$ in parallel.
The resulting multi-graph has at most $\lambda\UCd n\log n$ parallel edges between each pair of $u$ and $v$.
This takes $O(\log \log m+\log (\lambda n) )$ time and $O(rm)$ operations using $rm$ processors since $\sum_{\{u,v\} \in C^*(F)} |F^{uv}| \leq  rm$.

After (2), the resulting multi-graph has at most $\lambda\UCd n^3\log n$ edges (which is also bounded by $\sum_{\{u,v\} \in C^*(F)} |F^{uv}| \leq rm$).
Thus, step (3) can be done by calling the parallel spanner-construction algorithm $\lambda$ times, where each call takes $O(\log^2 n + \log\lambda)$ time using $rm$ processors by \cref{proposition:baswana-spanner}.
Therefore, step (3) takes $O(\lambda (\log^2 n + \log \lambda)) \lesssim O\prn*{\lambda \prn*{\log^2 n + \log\frac{1}{\eps}}}$ since $\lambda=O\prn*{\frac{r^3}{\eps^2}\log^6 m}$.

Summarizing the above discussion, steps (1)--(3) take $O\prn*{\log m+\lambda \prn*{\log^2 n + \log\frac{1}{\eps}}}$ time with $n^2m$ processors, which completes the analysis of  {\UHOnestep}$(H,\lambda)$.

We then bound the time complexity of {\UHSparsify}$(H,\eps)$. Our algorithm sequentially calls
{\UHOnestep}$(\tilde{H}_i,\lambda_i)$ for $i$ with $0\leq i\leq i_{\rm end}-1\leq T=\ceil*{ \log_{4/3}\frac{m}{m^*} }$.
Hence, the time complexity of our parallel algorithm is
$O\prn*{\sum_{i=0}^T \left( \log m_i+\lambda_i \prn*{\log^2 n + \log\frac{1}{\eps_i}} \right)}$,
which is $O\prn*{\frac{r^3}{\eps^2}\log^7 m \cdot \prn*{\log^2 n + \log\frac{1}{\eps}}}$ by $\lambda_i=O\prn*{\frac{r^3}{\eps^2}\log^6 m}$ and $T=O(\log m)$.

Finally, we use \cref{theorem:UHS-limitr} to deal with a general input hypergraph $H = (V, F, z)$ of rank $r$.
Suppose $F$ to be partitioned into $\set*{F_i}_{i=1}^{\ceil*{\log_2 r}}$, where each $f \in F_i$ satisfies $|f| \in  (2^{i-1}, 2^i]$ for $i = 1,\dots,\ceil*{\log_2 r}$.
We assign $|F_i|n^2$ processors to each hypergraph with the hyperedge set $F_i$, and compute $\eps\sqrt{\frac{2^{i-1}}{r}}$-spectral sparsifiers for $i = 1,\dots,\ceil*{\log_2 r}$ in parallel.
From the above discussion, each $\eps\sqrt{\frac{2^{i-1}}{r}}$-spectral sparsifier can be computed in
$O\prn*{\frac{2^{2i}r}{\eps^2}\log^7 |F_i| \cdot \prn*{\log^2 n + \log\frac{\sqrt{r}}{\eps}}} \lesssim O\prn*{\frac{r^3}{\eps^2}\log^7 m \cdot \prn*{\log^2 n + \log\frac{1}{\eps}}}$ time.
Therefore, the same asymptotic time complexity bound holds for general input hypergraphs of rank $r$.
\end{proof}

The above parallel algorithm uses $n^2 m$ processors only in the construction of $F^{uv}$, and the remaining operations use only $rm$ processors.
Moreover, we can readily construct $F^{uv}$ in each call to {\UHOnestep}$(\tilde{H}_i,\lambda_i)$ if we know $F^{uv}$ in the previous step, because $\tilde{H}_i$ is obtained from $\tilde{H}_{i-1}$ by removing unsampled hyperedges.
Using $rm_{i-1}$ processors, this can be done in $O(\log m_{i-1})$ time and with $O(rm_{i-1})$ operations by \cref{proposition:APC-kth-informal}.
This implies that we need $n^2 m$ processors only in the construction $F^{uv}$ at the beginning, and the number of processors can be reduced to $rm$ if the lists of $F^{uv}$ are given as input.
Under this assumption, the total number of operations, $O(rm)$, is work optimal.

\subsubsection{Fault Tolerance}\label{subsubsec:UHS_fault_tolerant}

Another advantage of the spanner-based sparsification is the fault-tolerance \citep{zhu2019improved}.
For constructing spanners, many fault-tolerant algorithms have been studied \citep{levcopoulos1998efficient,levcopoulos2002improved,chechik2010fault,dinitz2011fault,bodwin2018optimal}, and Zhu~et~al.~\citep{zhu2019improved} has recently applied the idea to graph sparsification.
Below are the definition of fault-tolerant graph sparsification in \citep{zhu2019improved} and its natural extension to hypergraphs.

\begin{definition}
    Let $G=(V,E,w)$ be a graph and $k$ be a positive integer.
    We say a graph-sparsification algorithm is \textit{weakly $k$-fault-tolerant} if for any $\eps \in (0,1)$, the algorithm outputs a re-weighted subgraph $\Gtl=(V,\Etl,\wtl)$ such that, for any edge deletion $J \subseteq E$ with $|J|\leq k$, $\Gtl'=(V,\Etl\setminus J,\wtl)$ is an $\eps$-spectral sparsifier of $G'=(V,E\setminus J,w)$ with high probability.
\end{definition}

\begin{definition}\label[definition]{definition:k-ft-hyper}
    Let $H=(V,F,z)$ be a hypergraph and $k$ be a positive integer.
    We say a hypergraph-sparsification algorithm is \textit{weakly $k$-fault-tolerant} if for any $\eps \in (0,1)$, the algorithm outputs a re-weighted sub-hypergraph $\Htl=(V,\Ftl,\ztl)$ such that, for any hyperedge deletion $J \subseteq F$ with $|J|\leq k$, $\Htl'=(V,\Ftl\setminus J,\ztl)$ is an $\eps$-spectral sparsifier of $H'=(V,F\setminus J,z)$ with high probability.
\end{definition}

In the above definitions, we use the term ``weakly'' to emphasize that the high-probability bounds are not uniform over all possible edge/hyperedge deletions; neither the result of \citep{zhu2019improved} nor ours offer such uniform bounds.
By contrast, fault-tolerant algorithms for constructing spanners usually enjoy uniform high-probability bounds over all possible edge deletions.
Indeed, combining our result and the analysis of~\citep{bansal2019new}, it would not difficult to obtain a uniformly $k$-fault-tolerant algorithm that returns an $\eps$-spectral sparsifier of size $O(k nr^3\log^2n/\eps^2)$, but the size bound is worse than that of \cref{theorem:FT}.
Improving this uniform bound would be an interesting future direction.

By extending the fault-tolerant spanner-based graph sparsification algorithm~\citep{zhu2019improved} to hypergraphs, we can obtain an algorithm with the following guarantee.
\begin{theorem}\label[theorem]{theorem:FT}
    There exists a weakly $k$-fault-tolerant algorithm for hypergraph sparsification such that for any hypergraph with $|V| = n$ the rank $r$ and $\eps\in (0,1)$, the output hypergraph has $O(kn + nr^3\log^2n/\eps^2)$ hyperedges with probability at least $1-O(1/n)$.
\end{theorem}

\begin{proof}
    Let $H = (V, F, z)$ be an input hypergraph with $|V| = n$ and $|F| = m$.
    Let $J \subseteq F$ be a hyperedge-deletion with $|J| \le k$.
    Following the notation used in {\UHSparsify}$(H, \eps)$, we let
    $(\Htl_i, \lambda_i) = ((V, \Ftl_{i}, \ztl_{i}), \lambda_i)$ and $\Htl_{i+1} = (V, \Ftl_{i+1}, \ztl_{i+1})$ be the input and output of {\UHOnestep}, respectively, in the $i$th iteration.
    We also let $m^* = nr^3\log^2n/\eps^2$.

    To achieve the fault tolerance, in each $i$th iteration, we select $\lambda_i + k$ disjoint hyperspanners, instead of $\lambda_i$ disjoint ones.
    Then, even if hyperedges in $J$ are deleted, there remain $\lambda_i$ hyperspanners.
    Therefore, as in \cref{subsec:UHS_proof_onestep}, we can show that $\Htl_{i+1}' = (V, \Ftl_{i+1} \setminus J, \ztl_{i+1})$ is an $\eps_i$-spectral sparsifier of $\Htl_{i}' = (V, \Ftl_{i} \setminus J, \ztl_{i})$.
    Consequently, the resulting $\Htl_{\Tend}'$ is an $\eps$-spectral sparsifier of $H' = (V, F\setminus J, z)$, as in the proof of \cref{theorem:UHS-main} in \cref{subsection:UHS_proof_iterative}.

    Below we analyze the increase in the size caused by selecting additional $k$ hyperspanners in each iteration.
    If $m_i = |\Htl_i| \le \UCc\frac{nr^3}{\eps^2}\log^2 n$ holds for some $i < T$, the output is already small as desired.
    If not, due to the presence of additional $k$ hyperspanners,
    $m_i$ decreases as $m_{i+1} \le \frac34 m_i + \UCa k n$ for $i = 0,\dots, T- 1$, as implied by \cref{lemma:UH-threequarter}.
    By solving this, we obtain
    \[
        m_T \le 4\UCa k n + \prn*{\frac34}^T \prn*{m - 4\UCa k n}
        \le
        4\UCa k n + \prn*{\frac34}^{\log_{4/3}\frac{m}{m^*}} m
        =
        O\prn*{kn + \frac{nr^3}{\eps^2}\log^2n}.
    \]
    Thus, the sparsifier size is bounded as stated in the theorem.
\end{proof}

Setting $r=2$, we obtain a size bound for ordinary graphs, which improves that of~\citep{zhu2019improved} by ${\rm poly}(\log n)$ factors.
This improvement is due to our adaptive choice of $\eps_i$ over the iterations and sharper analysis, whereas the $\eps$ value is fixed in the iterative algorithm of~\citep{zhu2019improved}.

\section*{Acknowledgements}
This work was supported by JST PRESTO Grant Number JPMJPR2126, JST ERATO Grant Number JPMJER1903, and JSPS KAKENHI Grant Number 20H05961.

\printbibliography[heading=bibintoc]

\end{document}

%% file: package.tex
\usepackage{amsmath,amsthm,amssymb,mathtools}

\usepackage{algorithm}

\usepackage[hidelinks,plainpages=false,colorlinks=true,citecolor=Navy,linkcolor=Maroon,urlcolor=Orchid]{hyperref} %hypertexnames=false
\usepackage[noend]{algpseudocode}

\algnewcommand{\Break}{\textbf{break}}

\usepackage{etoolbox}
\cslet{blx@noerroretextools}\empty%
\usepackage[date=year,eprint=false,doi=false,isbn=false,maxcitenames=2,natbib,url=false,sorting=nyt,style=ieee-alphabetic,backref=true]{biblatex}

\DeclareSourcemap{
    \maps[datatype=bibtex]{
        \map{
            \step[fieldset=editor, null]
        }
    }
}
\DefineBibliographyStrings{english}{%
  backrefpage = {cited on page},% originally "cited on page"
  backrefpages = {cited on pages},% originally "cited on pages"
}
\usepackage{cleveref}
\crefname{equation}{eq.}{eqs.}
\Crefname{equation}{Eq.}{Eqs.}
\crefname{algorithm}{Algorithm}{Algorithms}
\Crefname{algorithm}{Algorithm}{Algorithms}
\crefname{step}{Step}{Steps}
\Crefname{step}{Step}{Steps}
\crefname{theorem}{Theorem}{Theorems}
\Crefname{theorem}{Theorem}{Theorems}
\crefname{lemma}{Lemma}{Lemmas}
\Crefname{lemma}{Lemma}{Lemmas}
\crefname{corollary}{Corollary}{Corollaries}
\Crefname{corollary}{Corollary}{Corollaries}
\crefname{proposition}{Proposition}{Propositions}
\Crefname{proposition}{Proposition}{Propositions}
\crefname{assumption}{Assumption}{Assumptions}
\Crefname{assumption}{Assumption}{Assumptions}
\crefname{definition}{Definition}{Definitions}
\Crefname{definition}{Definition}{Definitions}
\crefname{figure}{Figure}{Figures}
\Crefname{figure}{Figure}{Figures}
\crefname{table}{Table}{Tables}
\Crefname{table}{Table}{Tables}
\crefname{section}{Section}{Sections}
\Crefname{section}{Section}{Sections}
\crefname{appendix}{Appendix}{Appendices}
\Crefname{appendix}{Appendix}{Appendices}

\usepackage{autonum}

\usepackage[utf8]{inputenc} % input from keybord is utf8
\usepackage[T1]{fontenc}    % use 8-bit T1 fonts
\usepackage{amsfonts}       % blackboard math symbols
\usepackage{nicefrac}       % compact symbols for 1/2, etc.
\usepackage{microtype}      % microtypography
\usepackage{bm}             % bold italic
\usepackage{bbm}

\usepackage{graphicx}
\usepackage[svgnames]{xcolor}
\usepackage{subcaption}
\usepackage{wrapfig}
\usepackage{svg}

\usepackage{tabularx}       % tabular with specified width
\usepackage{booktabs}       % professional-quality tables

\usepackage{thmtools}
\usepackage{url} % simple URL typesetting

\usepackage{braket}
\usepackage{mleftright}
\mleftright
\usepackage{multirow}
\usepackage{lipsum}
\allowdisplaybreaks
\usepackage{lscape}
\usepackage{tcolorbox}
\usepackage{xifthen}
\usepackage{xargs}
\usepackage{xspace}
\usepackage{enumerate}

\usepackage{tikz}
\usetikzlibrary{backgrounds}
\usetikzlibrary{arrows}
\usetikzlibrary{shapes,shapes.geometric,shapes.misc}

\tikzstyle{tikzfig}=[baseline=-0.25em,scale=0.5]

\pgfkeys{/tikz/tikzit fill/.initial=0}
\pgfkeys{/tikz/tikzit draw/.initial=0}
\pgfkeys{/tikz/tikzit shape/.initial=0}
\pgfkeys{/tikz/tikzit category/.initial=0}

\pgfdeclarelayer{edgelayer}
\pgfdeclarelayer{nodelayer}
\pgfsetlayers{background,edgelayer,nodelayer,main}

\tikzstyle{none}=[inner sep=0mm]

\tikzstyle{rectangle}=[fill=white, draw=black, shape=rectangle]
\tikzstyle{circle}=[fill=white, draw=black, shape=circle]
\tikzstyle{vertex}=[fill=black, draw=none, shape=circle]
\tikzstyle{textbox}=[fill=white, draw=none, shape=rectangle]

\tikzstyle{line}=[-, fill=none]
\tikzstyle{rightarrow}=[->]
\tikzstyle{leftarrow}=[fill=none, <-]

%% file: macro.tex
\newcommand{\Ccal}{\mathcal{C}}
\newcommand{\Dcal}{\mathcal{D}}

\newcommand{\Etl}{\tilde{E}}
\newcommand{\Ftl}{\tilde{F}}
\newcommand{\Gtl}{\tilde{G}}
\newcommand{\Htl}{\tilde{H}}

\newcommand{\Otl}{\tilde{O}}

\newcommand{\wtl}{\tilde{w}}

\newcommand{\ztl}{\tilde{z}}

\newcommand{\R}{\mathbb{R}}

\newcommand{\Z}{\mathbb{Z}}
\newcommand{\ones}{\bm{1}}

\newtheorem{theorem}{Theorem}[section]% Theorem
\newtheorem{lemma}[theorem]{Lemma}% Lemma
\newtheorem{proposition}[theorem]{Proposition}% Proposition
\newtheorem{corollary}[theorem]{Corollary}% Corollary
\theoremstyle{definition}
\newtheorem{definition}[theorem]{Definition}% Definition
\DeclareMathOperator*{\argmax}{argmax}
\DeclareMathOperator*{\argmin}{argmin}

\DeclareMathOperator{\poly}{poly}

\DeclareMathOperator{\E}{\mathbb{E}}

\DeclarePairedDelimiter\abs{\lvert}{\rvert}

\DeclarePairedDelimiter\ceil{\lceil}{\rceil}

\let\set\relax
\DeclarePairedDelimiter\set{\{}{\}}
\let\Set\relax
\DeclarePairedDelimiterX\Set[2]{\{}{\}}{\mspace{2mu}{#1}\;\delimsize|\;{#2}\mspace{2mu}}

\DeclarePairedDelimiterX\Brc[2]{[}{]}{\mspace{2mu}{#1}\;\delimsize|\;{#2}\mspace{2mu}}
\DeclarePairedDelimiter\prn{(}{)}
\DeclarePairedDelimiterX\Prn[2]{(}{)}{\mspace{2mu}{#1}\;\delimsize|\;{#2}\mspace{2mu}}

\newif\iffigure
\figurefalse